\documentclass[english,twocolumn,superscriptaddress,longbibliography,prl]{revtex4-2}
\usepackage{babel}
\usepackage[T1]{fontenc}
\usepackage[latin9]{inputenc}
\usepackage{verbatim}
\usepackage{amsmath}
\usepackage{amsthm}
\usepackage{amssymb}
\usepackage{graphicx}

\newtheorem{thm}{Theorem}

\newtheorem{prop}{Proposition}

\usepackage{physics}
\usepackage{bbm}
\usepackage{enumitem}
\usepackage{amsfonts}

\makeatother

\renewcommand{\theenumi}{(\roman{enumi})}

\providecommand{\ketbra}[1]{\ket{#1}\bra{#1}}

\newcommand{\ui}{\mathrm{i}}
\newcommand{\ue}{\mathrm{e}}

\begin{document}

%
%
%
%
%
%
%
%
%
%

\title{Genuine Multipartite Nonlocality Is Intrinsic to Quantum Networks}

\author{Patricia Contreras-Tejada}

\affiliation{Instituto de Ciencias Matemáticas, E-28049 Madrid, Spain}

\author{Carlos Palazuelos}

\affiliation{Departamento de Análisis Matemático y Matemática Aplicada, Universidad
Complutense de Madrid, E-28040 Madrid, Spain}

\affiliation{Instituto de Ciencias Matemáticas, E-28049 Madrid, Spain}

\author{Julio I. de Vicente}

\affiliation{Departamento de Matemáticas, Universidad Carlos III de Madrid, E-28911,
Leganés (Madrid), Spain}
\begin{abstract}
Quantum entanglement and nonlocality are inextricably linked. However,
while entanglement is necessary for nonlocality, it is not always
sufficient in the standard Bell scenario. We derive sufficient conditions
for entanglement to give rise to genuine multipartite nonlocality
in networks. We find that any network where the parties are connected
by bipartite pure entangled states is genuine multipartite nonlocal,
independently of the amount of entanglement in the shared states and
of the topology of the network. As an application of this result,
we also show that all pure genuine multipartite entangled states are
genuine multipartite nonlocal in the sense that measurements can be
found on finitely many copies of any genuine multipartite entangled
state to yield a genuine multipartite nonlocal behaviour. Our results
pave the way towards feasible manners of generating genuine multipartite
nonlocality using any connected network.
\end{abstract}
\maketitle
Correlations between quantum particles may be much stronger than those
between classical particles. Their applications are manifold:
cryptography \cite{gisin_quantum_2002,pirandola_advances_2020}, randomness
extraction, amplification and certification \cite{acin_certified_2016},
communication complexity reduction \cite{buhrman_nonlocality_2010},
etc., and the study of these \emph{nonlocal} correlations has led
to the growing field of device-independent quantum information processing
\cite{mayers_quantum_1998,acin_device-independent_2007,colbeck_quantum_2011}
(see also Ref. \cite{brunner_bell_2014}).

While bipartite nonlocality has been well researched in the past three decades, much less is known about the multipartite case. Still, correlations in quantum multicomponent systems have gained increasing attention recently, with applications in multiparty
cryptography \cite{aolita_fully_2012}, the understanding of condensed matter physics \cite{tura_detecting_2014,tura_energy_2017}, and  the development of quantum networks \cite{cavalcanti_quantum_2011,gisin_all_2017,supic_measurement-device-independent_2017,tavakoli_correlations_2017,renou_genuine_2019,
gisin_constraints_2020,krivachy_neural_2020}, particularly for quantum computation \cite{cirac_distributed_1999-1,howard_nonlocality_2012,howard_contextuality_2014} and correlating particles which never interacted \cite{branciard_characterizing_2010,branciard_bilocal_2012}.

A necessary condition for nonlocality is quantum entanglement.
Indeed, this is one reason why entangled states are useful
for communication-related tasks. However, not all entangled states
are nonlocal: some bipartite entangled states only yield local distributions \cite{werner_quantum_1989,barrett_nonsequential_2002}.
Still, for pure bipartite states, entanglement \textit{is} sufficient
for nonlocality, which is the content of Gisin's theorem \cite{gisin_bells_1991,gisin_maximal_1992}, and multipartite entangled pure states are never fully local \cite{popescu_generic_1992,gachechiladze_completing_2017}.  Interestingly, distributing certain bipartite entangled states
in certain multipartite networks yields  nonlocality even if the involved states are individually local \cite{sende_entanglement_2005,cavalcanti_quantum_2011,cavalcanti_nonlocality_2012,supic_measurement-device-independent_2017,luo_nonlocality_2018,luo_nonlocal_2019}.

Multipartite nonlocality is in principle harder to generate than bipartite nonlocality. By exploring the relationship between entanglement and nonlocality in the multipartite regime, in this Letter we show that networks simplify the job considerably: distributing arbitrarily low node-to-node entanglement is sufficient to observe truly multipartite nonlocal effects involving all parties in the network independently of its geometry. Added to its practical consequences for applications, this fact points to a deep property of quantum networks.

The multipartite setting has a richer structure than the
bipartite one, as different forms of entanglement and nonlocality
can be identified. Full separability (full locality) refers to systems
that do not display any form of entanglement (locality) whatsoever. However,
falsifying these models need not imply truly multipartite
quantum correlations since spreading them among two parties is sufficient.
Hence, a \textit{genuine multipartite} notion which inextricably
relates all parties together is more often considered. Here, a state
is genuine multipartite entangled (GME) if it is not a tensor product
of states of two subsets of parties, $M$ and its complement $\overline{M},$
i.e. of the form $\ket{\psi}=\ket{\psi_{M}}\otimes\ket{\psi_{\overline{M}}},$
or a convex combination of such states $\ketbra{\psi}$ across all
bipartitions. Analogously, a probability distribution $\left\{ P(\alpha_{1}\alpha_{2}...\alpha_{n}|\chi_{1}\chi_{2}...\chi_{n})\right\} _{\alpha_{1},...,\alpha_{n},\chi_{1},...,\chi_{n}}$
(with input $\chi_{i}$ and output $\alpha_{i}$ for party $i$) which
is not of the form 
\begin{equation}
\begin{aligned}P(\alpha_{1}\alpha_{2}...\alpha_{n} & |\chi_{1}\chi_{2}...\chi_{n})\\
=\sum_{M\subsetneq[n]}\sum_{\lambda} & q_{M}(\lambda)P_{M}(\{\alpha_{i}\}_{i\in M}|\{\chi_{i}\}_{i\in M}, \lambda)\\
 & \times P_{\overline{M}}(\{\alpha_{i}\}_{i\in\overline{M}}|\{\chi_{i}\}_{i\in\overline{M}}, \lambda),
\end{aligned}
\label{eq:BLdistrNpartite}
\end{equation}
where $q_{M}(\lambda)\geq0\,\forall\lambda,M,$ $\sum_{\lambda,M}q_{M}(\lambda)=1$
and $[n]:=\{1,...,n\},$ is genuine multipartite nonlocal (GMNL) \cite{svetlichny_distinguishing_1987,seevinck_bell-type_2002,collins_bell-type_2002},
and a state is GMNL if measurements giving rise to a GMNL
distribution exist.
The original definition \cite{svetlichny_distinguishing_1987} leaves the distributions $P_M, P_{\overline M}$ in equation (1) unrestricted; however, this has been shown to lead to operational problems \cite{buscemi_all_2012,gallego_operational_2012,bancal_definitions_2013,
geller_quantifying_2014,vicente_nonlocality_2014,
gallego_nonlocality_2017}.
Hence, like most recent works on the topic, we assume these distributions are nonsignalling, which captures most physical situations better
\cite{schmid_type-independent_2020,wolfe_quantifying_2020}. This means the marginal distributions for any subset of parties are independent of the inputs outside this subset, which is guaranteed if this holds for the marginals corresponding to ignoring just one party \cite{hanggi_device-independent_2010}:
\begin{equation}
\begin{aligned}
\sum_{\alpha_j} P_M (\{\alpha_i\}_{i\in M, i\neq j}, \alpha_j | \{\chi_i\}_{i\in M, i\neq j}, \chi_j, \lambda)=\\
\sum_{\alpha_j} P_M (\{\alpha_i\}_{i\in M, i\neq j}, \alpha_j | \{\chi_i\}_{i\in M, i\neq j}, \chi^{\prime}_j, \lambda),
\end{aligned}
\end{equation}
for all $\lambda, \chi_j\neq \chi^{\prime}_j$ and all parties $j,$ and similarly for $P_{\overline{M}}$.

In this Letter we show that the nonlocality arising from networks
of bipartite pure entangled states is a generic property and manifests
in its strongest form, GMNL. Specifically, we obtain that any connected
network of bipartite pure entangled states is GMNL. It was already
known that a star network of maximally entangled states is GMNL \cite{cavalcanti_quantum_2011},
but we provide a full, qualitative generalisation of this result by making it independent of both the amount of entanglement shared and the network topology. Thus, we show GMNL is an intrinsic property of networks of pure bipartite entangled states.

Further, there are known mixed GME states that are not GMNL \cite{augusiak_entanglement_2015,augusiak_constructions_2018}---some
are even fully local \cite{bowles_genuinely_2016}. Still, it is not known whether Gisin\textquoteright s theorem extends to
the genuine multipartite regime. Recent results
show that, for pure $n$-qubit symmetric states \cite{chen_test_2014}
and all pure 3-qubit states \cite{yu_tripartite_2013}, GME implies
GMNL (at the single-copy level) \footnote{Hidden GMNL for three parties beyond qubits can be shown if some form of preprocessing is allowed.}. Our result above shows that all pure
GME states that have a network structure are GMNL; interestingly,
we further apply this property to establish a second result: all pure GME states are GMNL in the sense that measurements
can be found on finitely many copies of any GME state to yield a GMNL
behaviour. We thus tighten the relationship between multipartite entanglement and nonlocality.

Our construction exploits the fact that the set of non-GME states
is not closed under tensor products, i.e. GME can be superactivated
by taking tensor products of states that are unentangled across different
bipartitions. Thus, GME can be achieved by distributing bipartite
entangled states among different pairs of parties. To obtain our results,
we extend the superactivation property \cite{navascues_activation_2011,palazuelos_super-activation_2012,caban_activation_2015}
from the level of states to that of probability distributions, i.e.
GMNL can be superactivated by taking Cartesian products of probability
distributions that are local across different bipartitions. In fact,
when considering copies of quantum states, we only consider local
measurements performed on each copy separately, thus pointing at a
stronger notion of superactivation to achieve GMNL.

\paragraph{Definitions and preliminaries}

We consider distributions arising from GME states, and ask whether
they satisfy (\ref{eq:BLdistrNpartite}). The set of
distributions of the form (\ref{eq:BLdistrNpartite}) is a polytope: indeed, the set of local distributions
across each bipartition $M|\overline{M}$ is a polytope, and
convex combinations preserve that structure. We call this $n$-partite
polytope $\mathcal{B}_{n}.$ We call an inequality 
\begin{equation}
\sum_{\substack{\alpha_{i}\chi_{i}\\
i\in[n]
}
}c_{\alpha_{1}...\alpha_{n}\chi_{1},...,\chi_{n}}P(\alpha_{1}...\alpha_{n}|\chi_{1}...\chi_{n})\leq c_{0}
\end{equation}
which holds for all $P$ of the form (\ref{eq:BLdistrNpartite})
a GMNL inequality. 

We use results from Ref. \cite{pironio_lifting_2005} to lift inequalities
to account for more parties, inputs and outputs. They consider the fully local polytope $\mathcal{L}$,
which only includes distributions
\begin{equation}
P(\alpha\beta|\chi\upsilon)=\sum_{\lambda}q(\lambda)P_{A}(\alpha|\chi,\lambda)P_{B}(\beta|\upsilon,\lambda)\label{eq:FLdistr}
\end{equation}
where each party may have different numbers of inputs and outputs
(more parties may be considered by adding more distributions correlated
only by $\lambda$). Polytope $\mathcal{B}_{n}$ includes convex combinations of distributions that are local across different bipartitions $M|\overline{M}$ of the parties, but the lifting
results in \cite{pironio_lifting_2005} still hold. Indeed, to check an inequality holds for a polytope, it is sufficient by convexity
to check the extremal points. As all extremal points
in $\mathcal{B}_{n}$ are contained in some polytope $\mathcal{L}$,
lifting results for $\mathcal{L}$ can be straightforwardly extended
to $\mathcal{B}_{n}$.

We also use the EPR2 decomposition \cite{elitzur_quantum_1992} and
its multipartite extension \cite{almeida_multipartite_2010}: any
distribution $P$ can be expressed (nonuniquely) as
\begin{equation}
\begin{aligned}P(\alpha_{1}... \alpha_{n}|\chi_{1}...\chi_{n})
=&\sum_{M\subsetneq[n]}  p_{L}^{M}P_{L}^{M}(\alpha_{1}...\alpha_{n}|\chi_{1}...\chi_{n})\\
 & +p_{NS}P_{NS}(\alpha_{1}...\alpha_{n}|\chi_{1}...\chi_{n})
\end{aligned}
\label{eq:EPR2-multipartite}
\end{equation}
where $\sum_{M\subsetneq[n]}p_{L}^{M}+p_{NS}=1,$ $P_{L}^{M}$ is
local across $M|\overline{M}$ (i.e. satisfies equation (\ref{eq:FLdistr})
with parties grouped as per $M|\overline{M}$), and $P_{NS}$ is nonsignalling.
$P$ is GMNL if all such decompositions have $p_{NS}>0,$ and fully GMNL
if all such decompositions have $p_{NS}=1.$ A state $\rho$ is fully GMNL
if, $\forall \varepsilon>0,$ there exist local measurements giving rise to some $P$ such that any decomposition (\ref{eq:EPR2-multipartite})
has $p_{NS}>1-\varepsilon.$ Bipartite distributions and states may
be nonlocal or fully nonlocal \footnote{Not to be confused with ``nonfully local'', which is the opposite
of ``fully local''. ``fully nonlocal'' is a particular case of
``nonfully local''.} analogously.

\paragraph{GMNL from bipartite entanglement}

Our first result shows that any connected network of pure bipartite
entanglement (see Figure \ref{fig:GMNL-from-bipartite}) is GMNL.

\begin{figure}
\includegraphics[width=0.5\textwidth]{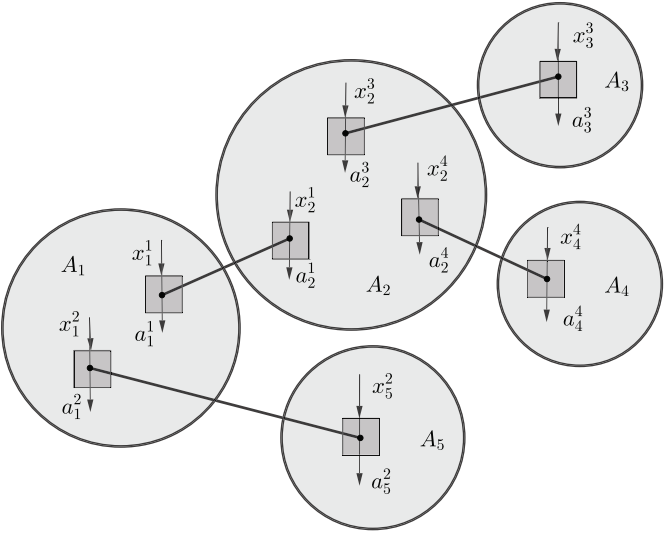}

\caption{Connected network of bipartite entanglement. For each $i\in[n],$
party $A_{i}$ has input $x_{i}^{k}$ and output $a_{i}^{k}$ on the
particle at edge $k.$ Particles connected by an edge are entangled. 
\label{fig:GMNL-from-bipartite}}
\end{figure}

\begin{thm}
\label{thm:gmnl-from-bipartite-ent} Any connected network of bipartite
pure entangled states is GMNL.
\end{thm}

We now outline the proof for a tripartite network
where $A_{1}$ is entangled to each of $A_{2}$ and $A_{3}$, and
leave the general case to \cite{supmat}. Since it turns out sufficient to measure individually on each party's different
particles (see Figure \ref{fig:GMNL-from-bipartite} for the $n$-partite
structure), the shared distribution $P(a_{1}^{1}a_{1}^{2},a_{2}^{1},a_{3}^{2}|x_{1}^{1}x_{1}^{2},x_{2}^{1},x_{3}^{2})$
takes the form
\begin{equation}
P_{1}(a_{1}^{1}a_{2}^{1}|x_{1}^{1}x_{2}^{1})P_{2}(a_{1}^{2}a_{3}^{2}|x_{1}^{2}x_{3}^{2})\label{eq:P3partite}
\end{equation}
where parties $A_{i},A_{j}$ are connected by edge $k$ (we label vertices and edges independently), and $P_{k}(a_{i}^{k}a_{j}^{k}|x_{i}^{k}x_{j}^{k})$
is the distribution arising from the state at edge $k.$

The proof considers three cases, depending on whether the shared states are maximally entangled.
If none are, we devise inequalities to detect bipartite nonlocality at each edge of the network, and combine them to form a multipartite inequality. Then, we find measurements on the shared states to violate it. If both states are maximally entangled, existing results show the network is fully GMNL \cite{cavalcanti_quantum_2011,almeida_multipartite_2010}. Combining these two cases for a heterogeneous network completes the proof.

To prove the first case, we take bipartite inequalities between $A_1$ and each other party, lift them to three parties and combine them using Refs. \cite{pironio_lifting_2005,curchod_versatile_2019}, to obtain the following GMNL inequality:
\begin{equation}
\begin{aligned}I_{3} & =I^{1}+I^{2}+P(00,0,0|00,0,0)\\
 & -\sum_{a_{1}^{2}=0,1}P(0a_{1}^{2},0,0|00,0,0)\\
 & -\sum_{a_{1}^{1}=0,1}P(a_{1}^{1}0,0,0|00,0,0)\leq0\,.
\end{aligned}
\label{eq:Iabc44}
\end{equation}
Here,
\begin{equation}
\begin{aligned}I^{1}= & \sum_{a_{1}^{2}=0,1}\left[P(0a_{1}^{2},0,0|00,0,0)-P(0a_{1}^{2},1,0|00,1,0)\right.\\
- & \left.P(1a_{1}^{2},0,0|10,0,0)-P(0a_{1}^{2},0,0|10,1,0)\right]\leq0;
\end{aligned}
\label{eq:IAB44}
\end{equation}
\begin{equation}
\begin{aligned}I^{2}= & \sum_{a_{1}^{1}=0,1}\left[P(a_{1}^{1}0,0,0|00,0,0)-P(a_{1}^{1}0,0,1|00,0,1)\right.\\
- & \left.P(a_{1}^{1}1,0,0|01,0,0)-P(a_{1}^{1}0,0,0|01,0,1)\right]\leq0
\end{aligned}
\label{eq:IAC44}
\end{equation}
are liftings of 
\begin{equation}
I=P(00|00)-P(01|01)-P(10|10)-P(00|11)\leq0\label{eq:CHSHseed}
\end{equation}
to three parties with $A_{1}$ having 4 inputs and 4 outputs. Inequality
(\ref{eq:CHSHseed}) is equivalent to the CHSH inequality \cite{clauser_proposed_1969}
for nonsignalling distributions \cite{curchod_versatile_2019}. Thus,
inequalities (\ref{eq:IAB44}), (\ref{eq:IAC44}) are satisfied by
distributions that are local across $A_{1}|A_{2}$ and $A_{1}|A_{3}$
respectively. To see that equation (\ref{eq:Iabc44}) is a GMNL inequality
it is sufficient to check it holds for distributions that are local
across some bipartition. This is straightforwardly done by observing
the cancellations that occur when $I^{1}$ or $I^{2}$
are $\leq0.$

Since both states are less-than-maximally entangled, $A_{1}$ can
satisfy Hardy's paradox \cite{hardy_quantum_1992,hardy_nonlocality_1993}
with each other party, achieving
\begin{equation}
P_{k}(00|00)>0=P_{k}(01|01)=P_{k}(10|10)=P_{k}(00|11)
\end{equation}
for both $k$ (the proof for qubits in Refs. \cite{hardy_quantum_1992,hardy_nonlocality_1993}
is extended to qudits by measuring on a two-dimensional subspace,
see Ref. \cite{supmat}). Then, each negative term in $I^{1}$ and $I^{2}$
is zero, as
\begin{equation}
\begin{aligned}\sum_{a_{1}^{2}=0,1}P(0a_{1}^{2},1,0|00,1,0) & =P_{1}(01|01)\sum_{a_{1}^{2}=0,1}P_{2}(a_{1}^{2}0|00)\end{aligned}
\end{equation}
and similarly for the others. Hence, only
\begin{equation}
P(00,0,0|00,0,0)=P_{1}(00|00)P_{2}(00|00)>0
\end{equation}
survives, violating the inequality.

If, instead, $A_{1} A_2$ share a maximally entangled state,
and $A_2 A_3$ share a less-than-maximally entangled state, then $A_{1}A_{3}$
can measure so that $P_{2}$ satisfies Hardy's paradox;
hence $\exists\,\varepsilon>0$ such that its local component
in any EPR2 decomposition satisfies
\begin{equation}
p_{L,2}\leq1-\varepsilon.\label{eq:plocal-Hardy}
\end{equation}
Since the maximally entangled state is fully nonlocal \cite{barrett_maximally_2006},
for this $\varepsilon,$ $A_{1}A_{2}$ can measure such
that any EPR2 decomposition of $P_{1}$ satisfies
\begin{equation}
p_{L,1}<\varepsilon.\label{eq:plocal-maxent}
\end{equation}

Then, we assume for a contradiction that $P(a_{1}^{1}a_{1}^{2},a_{2}^{1},a_{3}^{2}|x_{1}^{1}x_{1}^{2},x_{2}^{1},x_{3}^{2})$
is not GMNL and decompose it in its bipartite splittings,
\begin{equation}
\begin{aligned}P(a_{1}^{1}a_{1}^{2}, & a_{2}^{1},a_{3}^{2}|x_{1}^{1}x_{1}^{2},x_{2}^{1},x_{3}^{2})\\
=\sum_{\lambda} & \left(p_{L}(\lambda)P_{A_{1}A_{2}}(a_{1}^{1}a_{1}^{2},a_{2}^{1}|x_{1}^{1}x_{1}^{2},x_{2}^{1},\lambda)P_{A_{3}}(a_{3}^{2}|x_{3}^{2},\lambda)\right.\\
 & +q_{L}(\lambda)P_{A_{1}A_{3}}(a_{1}^{1}a_{1}^{2},a_{3}^{2}|x_{1}^{1}x_{1}^{2},x_{3}^{2},\lambda)P_{A_{2}}(a_{2}^{1}|x_{2}^{1},\lambda)\\
 & \left.+r_{L}(\lambda)P_{A_{1}}(a_{1}^{1}a_{1}^{2}|x_{1}^{1}x_{1}^{2},\lambda)P_{A_{2}A_{3}}(a_{2}^{1},a_{3}^{2}|x_{2}^{1},x_{3}^{2},\lambda)\right)
\end{aligned}
\label{eq:PmixBL}
\end{equation}
where $\sum_{\lambda}\left[p_{L}(\lambda)+q_{L}(\lambda)+r_{L}(\lambda)\right]=1.$

Summing equation (\ref{eq:PmixBL}) over $a_{1}^{2},a_{3}^{2}$ and
using equation (\ref{eq:P3partite}), we get an EPR2 decomposition
of $P_{1}$ with local components $q_{L},r_{L}.$ By equation (\ref{eq:plocal-maxent}),
this entails $\sum_{\lambda}\left[q_{L}(\lambda)+r_{L}(\lambda)\right]<\varepsilon,$
so 
\begin{equation}
\sum_{\lambda}p_{L}(\lambda)>1-\varepsilon.
\end{equation}

Summing, instead, equation (\ref{eq:PmixBL}) over $a_{1}^{1},a_{2}^{1},$
we obtain an EPR2 decomposition of $P_{2}$ whose only nonnegligible
component, $\sum_{\lambda}p_{L}(\lambda),$ is local in $A_{1}|A_{3},$
contradicting equation (\ref{eq:plocal-Hardy}). Therefore, $P$ must
be GMNL.

\paragraph{GMNL from GME}

By Theorem \ref{thm:gmnl-from-bipartite-ent}, a star network whose central node shares pure entanglement with all others is GMNL.
We now ask whether all GME states are GMNL (i.e.
the genuine multipartite extension of Gisin's theorem). We show $(n-1)$
copies of any pure GME $n$-partite state suffice to generate $n$-partite
GMNL. We do this by generating a distribution from these copies that
mimics the star network configuration.
\begin{figure}
\includegraphics[width=0.5\textwidth]{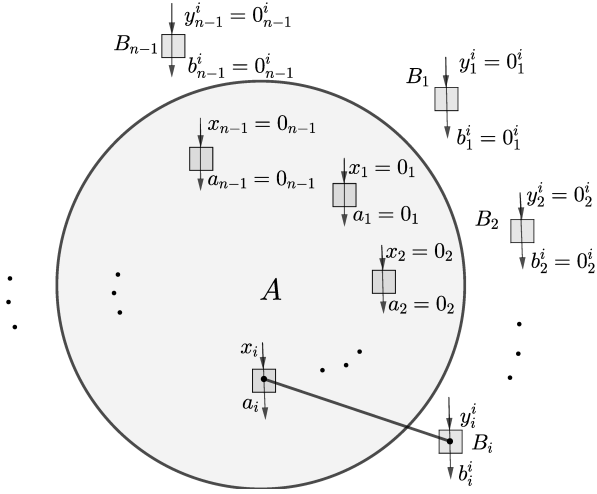}

\caption{Element $i\in[n-1]$ of the star network of bipartite entanglement
created from a GME state $\ket{\Psi}.$ Parties $\{B_{j}\}_{j\in[n-1],j\protect\neq i}$
have already measured $\ket{\Psi}$ and are left unentangled. Alice
and party $B_{i}$ share a pure bipartite entangled state. Alice has
input $x_{i}$ and output $a_{i}$ while each party $B_{j},$ $j\in[n-1],$
has input $y_{j}^{i}$ and output $b_{j}^{i}.$ \label{fig:GMNL-from-GME}}

\end{figure}

\begin{thm}
\label{thm:copies}Any GME state $\ket{\Psi}\in\mathcal{H}_{1}\otimes...\otimes\mathcal{H}_{n}\cong(\mathbb{C}^{d})^{\otimes n}$
is such that $\ket{\Psi}^{\otimes(n-1)}$ is GMNL. 
\end{thm}

The full proof is given in \cite{supmat}, and we presently outline
the tripartite case. Hence, we consider two copies of the state.
For each copy, we derive measurements for Bob1 and Bob2 that leave
Alice bipartitely entangled with Bob2 and Bob1 respectively. This
yields a network as in equation (\ref{eq:P3partite}) but postselected
on the inputs and outputs of these measurements. We generalise
Theorem \ref{thm:gmnl-from-bipartite-ent} to show this network is
also GMNL.

For $i,j=1,2,$ on copy $i,$ $B_{j}$'s measurements have input $y_{j}^{i}$
and output $b_{j}^{i}$ and Alice's measurement has input $x_{i}$
and output $a_{i}.$ We denote $B_{j}$'s inputs and outputs in terms
of their digits as $\upsilon_{j}=y_{j}^{1}y_{j}^{2}$ and $\beta_{j}=b_{j}^{1}b_{j}^{2}.$
Then, after measurement, the parties share a distribution
\begin{equation}
\begin{aligned}P(\alpha & \beta_{1}\beta_{2}|\chi\upsilon_{1}\upsilon_{2})\\
 & =P_{1}(a_{1},b_{1}^{1}b_{2}^{1}|x_{1},y_{1}^{1}y_{2}^{1})P_{2}(a_{2},b_{1}^{2}b_{2}^{2}|x_{2},y_{1}^{2}y_{2}^{2})\,.
\end{aligned}
\label{eq:Pnpartite-fromGME}
\end{equation}

For each $i,j=1,2,$ $i\neq j,$ we assume $B_{j}$ uses input
$0_{j}^{i}$ and output $0_{j}^{i}$ to project the $i$th copy of
$\ket{\Psi}$ onto $\ket{\phi_{i}}_{AB_{i}},$ as shown in Figure
\ref{fig:GMNL-from-GME} for $n$ parties. Then, Refs. \cite{popescu_generic_1992,gachechiladze_completing_2017}
and a continuity argument serve to show we only have two possibilities
for each $i$: either there exists an input and output per party such
that $\ket{\phi_{i}}_{AB_{i}}$ is less-than-maximally entangled,
or there exists an input per party such that, for all outputs, $\ket{\phi_{i}}_{AB_{i}}$
is maximally entangled. In each case we generalise the proof in Theorem
\ref{thm:gmnl-from-bipartite-ent} to show $\ket{\Psi}^{\otimes2}$
is GMNL.

If both $\ket{\phi_{i}}_{AB_{i}},$ $i=1,2$ are less-than-maximally
entangled, we use the following expression, which is a GMNL inequality
by the same reasoning as in Theorem \ref{thm:gmnl-from-bipartite-ent}:
\begin{equation}
\begin{aligned}I_{3} & =\sum_{i=1}^{2}I^{i}+P(00,00,00|00,00,00)\\
 & -\sum_{i=1}^{2}\sum_{\substack{a_{j},b_{i}^{j},\\
b_{j}^{j}=0,1,\\
j\neq i
}
}P(0_{i}a_{j}\,,0_{i}^{i}b_{i}^{j}\,,0_{j}^{i}b_{j}^{j}\,|0_{i}0_{j}\,,0_{i}^{i}0_{i}^{j}\,,0_{j}^{i}0_{j}^{j})\leq0,
\end{aligned}
\end{equation}
where
\begin{equation}
\begin{aligned}I^{i}= & \sum_{\substack{a_{j},b_{i}^{j},b_{j}^{j}=0,1,\\
j\neq i
}
}\left[P(0_{i}a_{j}\,,0_{i}^{i}b_{i}^{j}\,,0_{j}^{i}b_{j}^{j}\,|0_{i}0_{j}\,,0_{i}^{i}0_{i}^{j}\,,0_{j}^{i}0_{j}^{j})\right.\\
 & -P(0_{i}a_{j}\,,1_{i}^{i}b_{i}^{j}\,,0_{j}^{i}b_{j}^{j}\,|0_{i}0_{j}\,,1_{i}^{i}0_{i}^{j}\,,0_{j}^{i}0_{j}^{j})\\
 & -P(1_{i}a_{j}\,,0_{i}^{i}b_{i}^{j}\,,0_{j}^{i}b_{j}^{j}\,|1_{i}0_{j}\,,0_{i}^{i}0_{i}^{j}\,,0_{j}^{i}0_{j}^{j})\\
 & \left.-P(0_{i}a_{j}\,,0_{i}^{i}b_{i}^{j}\,,0_{j}^{i}b_{j}^{j}\,|1_{i}0_{j}\,,1_{i}^{i}0_{i}^{j}\,,0_{j}^{i}0_{j}^{j})\right]\,.
\end{aligned}
\end{equation}

Evaluating the inequality on the distribution (\ref{eq:Pnpartite-fromGME}),
we find again that all negative terms in each $I^{i}$ can be sent
to zero. For each $i$ we get, for example,
\begin{equation}
\begin{aligned}\sum_{\substack{a_{j},b_{i}^{j},b_{j}^{j}\\
=0,1
}
} & P(0_{i}a_{j}\,,1_{i}^{i}b_{i}^{j}\,,0_{j}^{i}b_{j}^{j}\,|0_{i}0_{j}\,,1_{i}^{i}0_{i}^{j}\,,0_{j}^{i}0_{j}^{j})\\
 & =P_{i}(0_{i}1_{i}^{i}0_{j}^{i}|0_{i}1_{i}^{i}0_{j}^{i})
\end{aligned}
\end{equation}
as the sum over $P_{j}$ is 1. But, conditioned on $B_{j}$'s input
and output being $0_{j}^{i},$ parties $AB_{i}$ can measure so
$P_{i}$ satisfies Hardy's paradox, hence this term is zero, and similarly
for the other two negative terms. This means all terms in $I_{3}$
are zero except $P(00,00,00|00,00,00)>0,$ violating the inequality.
Therefore, $\ket{\Psi}^{\otimes2}$ is GMNL.

If, for both $i=1,2,$ there exists a local measurement for party
$B_{j},$ $j\neq i$ such that, for \emph{all }outputs, $\ket{\Psi}$
is projected onto a maximally entangled state $\ket{\phi_{i}}_{AB_{i}},$
then $\ket{\Psi}$ satisfies Theorem 2 in Ref. \cite{almeida_multipartite_2010},
so $\ket{\Psi}$ itself is GMNL. Therefore so is $\ket{\Psi}^{\otimes2}.$

Finally, if $\ket{\phi_{1}}_{AB_{1}}$ is maximally entangled for
all of $B_{2}$'s outputs, and $\ket{\phi_{2}}_{AB_{2}}$ is less-than-maximally
entangled, using Refs. \cite{popescu_generic_1992,almeida_multipartite_2010}
we deduce that the bipartite EPR2 components of
$P_{1,2}$ across $A|B_{1,2}$ respectively are bounded like
in Theorem \ref{thm:gmnl-from-bipartite-ent}. That is,
$\exists\, \varepsilon>0$ such that the local component of any EPR2 decomposition
across $A|B_{2}$ satisfies
\begin{equation}
p_{L,2}^{A|B_{2}}\leq1-\varepsilon\label{eq:pHlocal-copies-3parties}
\end{equation}
and, given this $\varepsilon,$ parties $AB_{1}$ can measure locally
such that all bipartite EPR2 decompositions across $A|B_{1}$ have
a local component
\begin{equation}
p_{L,1}^{A|B_{1}}<\varepsilon.\label{eq:Plocal-maxent-copies-3parties}
\end{equation}

Then, we assume $P(\alpha\beta_{1}\beta_{2}|\chi\upsilon_{1}\upsilon_{2})$
is not GMNL and decompose it in local terms across different bipartitions,
like in equation (\ref{eq:PmixBL}) in Theorem \ref{thm:gmnl-from-bipartite-ent}.
Summing over $a_{2},b_{j}^{2},$ $j=1,2$ gives an EPR2 decomposition
of $P_{1}$ whose local components can be bounded using equation (\ref{eq:Plocal-maxent-copies-3parties}).
Summing over $a_{1},b_{j}^{1},$ $j=1,2$ instead gives an EPR2 decomposition
of $P_{2}.$ But the bound on the local component of $P_{1}$ entails
a bound on that of $P_{2}$ which contradicts equation (\ref{eq:pHlocal-copies-3parties}),
proving $P$ is GMNL.

\paragraph{Conclusions}

We have shown that GMNL can be obtained by distributing arbitrary
pure bipartite entanglement, which paves the way towards feasible
generation of GMNL from any network. In fact, our results
imply that, given a set of nodes, distributing entanglement
in the form of a tree is sufficient to observe GMNL. In practical applications,
the entanglement shared by the nodes would unavoidably degrade to
mixed-state form. By continuity, the GMNL in the networks of pure
bipartite entanglement considered here must be robust to some noise.
Quantifying this tolerance is interesting for future work.

Further, we have shown that a tensor product of finitely many GME
states is always GMNL. The question whether all single-copy pure
GME states are GMNL remains open.

The assumption that the distributions $P_M,\,P_{\overline{M}}$ are nonsignalling in the GMNL definition is physically natural. Still, removing it raises the stakes to achieve nonlocality, and establishing analogous results with the stronger definition is an open question.

Very recently, Ref. \cite{navascues_genuine_2020} proposed the concept
of ``genuine network entanglement'', a stricter notion than GME
which rules out states which are a tensor product of non-GME states.
One might hope that states that are GME but not genuine network entangled
might be detected device independently by not passing GMNL tests.
However, our results show this will not work. Any distribution of
pure bipartite states, even with arbitrarily weak entanglement, always
displays GMNL as long as all parties are connected. This further motivates
searching for an analogous concept of genuine network nonlocality
that may detect genuine network entanglement.

\begin{acknowledgments}
The authors thank Aleksander M. Kubicki and Elie Wolfe for enlightening
discussions. This research was funded by the Spanish MINECO through
Grant No. MTM2017-88385-P and by the Comunidad de Madrid through Grant No. QUITEMAD-CMS2018/TCS-4342. We also acknowledge funding from the Spanish MINECO, through the ``Severo Ochoa Programme
for Centres of Excellence in R\&D'' SEV-2015-0554 and from the Spanish
National Research Council, through the ``Ayuda extraordinaria a Centros
de Excelencia Severo Ochoa'' 20205CEX001 (PCT and CP) and the Spanish
MINECO Grant No. MTM2017-84098-P (JIdV). PCT is grateful for the hospitality
of Perimeter Institute where part of this work was carried out. Research
at Perimeter Institute is supported in part by the Government of Canada
through the Department of Innovation, Science and Economic Development
Canada and by the Province of Ontario through the Ministry of Economic
Development, Job Creation and Trade.
\end{acknowledgments}

%

\cleardoublepage1
\onecolumngrid
\setcounter{thm}{0}
\setcounter{lem}{0}
\setcounter{prop}{0}
\setcounter{equation}{0}
\begin{center}
\large{\textbf{Supplemental material}}
\end{center}
\appendix
%
%
%
%
%
%
%
%
%
%
%
%
%
%
%
%
%
%
%
%
%
%

\title{Supplemental material}

\maketitle
We prove Theorems \ref{thm:gmnl-from-bipartite-ent} and \ref{thm:copies}
in the main text. For the reader's convenience, we also restate some
definitions and the results.

We say a probability distribution $\left\{ P(\alpha_{1}\alpha_{2}...\alpha_{n}|\chi_{1}\chi_{2}...\chi_{n})\right\} _{\alpha_{1},...,\alpha_{n},\chi_{1},...,\chi_{n}}$
(with input $\chi_{i}$ and output $\alpha_{i}$ for party $i$) is
genuine multipartite nonlocal (GMNL) if it cannot be written in the
form 
\begin{equation}
P(\alpha_{1}\alpha_{2}...\alpha_{n}|\chi_{1}\chi_{2}...\chi_{n})=\sum_{M}\sum_{\lambda}q_{M}(\lambda)P_{M}(\{\alpha_{i}\}_{i\in M}|\{\chi_{i}\}_{i\in M},\lambda)P_{\overline{M}}(\{\alpha_{i}\}_{i\in\overline{M}}|\{\chi_{i}\}_{i\in\overline{M}},\lambda),\label{eq:BLdistrNpartite}
\end{equation}
where here and in the rest of the paper we take $\emptyset\neq M\subsetneq[n],$
and we ignore duplicate bipartitions (e.g. by assuming $M$ always
contains party $A_{1}$), $q_{M}(\lambda)\geq0$ for each $\lambda,M,$
$\sum_{\lambda,M}q_{M}(\lambda)=1$ and the distributions $P_{M},P_{\overline{M}}$
on each bipartition are nonsignalling. We will use the notation $[n]:=\{1,...,n\}$
throughout. A state is GMNL if there exist measurements which give
rise to a distribution that cannot be written as (\ref{eq:BLdistrNpartite}).

Given an inequality 
\begin{equation}
\sum_{\substack{\alpha_{i}\chi_{i}\\
i\in[n]
}
}c_{\alpha_{1}...\alpha_{n}\chi_{1},...,\chi_{n}}P(\alpha_{1}...\alpha_{n}|\chi_{1}...\chi_{n})\leq c_{0}
\end{equation}
which holds for all distributions $P$ of the form (\ref{eq:BLdistrNpartite}),
we refer to it as a GMNL inequality. The set of points $P(\alpha_{1}...\alpha_{n}|\chi_{1}..\chi_{n})$
which satisfy it with equality is a \emph{face }of the polytope $\mathcal{B}_{n}$
of $n$-partite distributions (\ref{eq:BLdistrNpartite}). Inequalities
are said to \emph{support} faces of the polytope. Faces $F\neq\mathcal{B}_{n}$
of maximal dimension are \emph{facets}, and inequalities which support
facets are called \emph{facet inequalities}.

The multipartite EPR2 decomposition \cite{elitzur_quantum_1992,almeida_multipartite_2010}
of a probability distribution $P$ is 
\begin{equation}
\begin{aligned}P(\alpha_{1} & ...\alpha_{n}|\chi_{1}...\chi_{n})\\
 & =\sum_{M}p_{L}^{M}P_{L}^{M}(\alpha_{1}...\alpha_{n}|\chi_{1}...\chi_{n})+p_{NS}P_{NS}(\alpha_{1}...\alpha_{n}|\chi_{1}...\chi_{n}),
\end{aligned}
\label{eq:EPR2-multipartite}
\end{equation}
where $p_{L}^{M}\geq0$ for every $M,$ $p_{NS}\geq0$ and 
\begin{equation}
\sum_{M}p_{L}^{M}+p_{NS}=1,
\end{equation}
$P_{L}^{M}$ is local across the bipartition $M|\overline{M},$ and
$P_{NS}$ is nonsignalling. We are interested in decompositions which
maximise the local EPR2 components, in order to deduce properties
about the distributions. For a distribution $P,$ we define 
\begin{equation}
EPR2(P)=\max\left\{ \sum_{M}p_{L}^{M}:P=\sum_{M}p_{L}^{M}P_{L}^{M}+p_{NS}P_{NS},\sum_{M}p_{L}^{M}+p_{NS}=1\right\} 
\end{equation}
and, for a state $\rho,$ we define (with a slight abuse of notation)
\begin{equation}
EPR2(\rho)=\inf\left\{ EPR2(P):P=\tr\left(\bigotimes_{i=1}^{n}E_{\alpha_{i}|\chi_{i}}^{i}\rho\right)\right\} ,
\end{equation}
where the infimum is taken over local measurements $E_{\alpha_{i}|\chi_{i}}^{i}$
on each particle such that 
\begin{equation}
E_{\alpha_{i}|\chi_{i}}^{i}\succcurlyeq0\:\forall\alpha_{i},\chi_{i},\:\sum_{\alpha_{i}}E_{\alpha_{i}|\chi_{i}}^{i}=1\:\forall\chi_{i},\:\forall i\in[n],
\end{equation}
with any number of inputs and outputs. Then, a distribution $P$ or
a state $\rho$ are GMNL if $EPR2(\cdot)<1,$ while they are fully GMNL
if $EPR2(\cdot)=0.$ When considering bipartite distributions and
states, the analogous property is termed full-nonlocality. Notice
that the optimisation for probability distributions yields a maximum
since the number of inputs and outputs is fixed. Instead, the optimisation
for a state may involve measurements with an arbitrarily large number
of inputs or outputs, as is the case for the maximally entangled state
\cite{barrett_maximally_2006}. In this work, the number of inputs
and outputs is always finite, and this will become relevant when bounding
the EPR2 components of distributions arising from maximally entangled
states in Theorems \ref{thm:gmnl-from-bipartite-ent} and \ref{thm:copies}.\medskip{}

\paragraph{GMNL from bipartite entanglement}
\begin{thm}
\label{thm:gmnl-from-bipartite-ent}Any connected network of bipartite
pure entangled states is GMNL.
\end{thm}

\begin{proof}
We consider the network as a connected graph where vertices are parties
and edges are states. The graph is such that, at each vertex, there
is one particle for every incident edge \footnote{Throughout the proof we assume $k\geq2.$ If $k=1,$ there are only
two parties sharing bipartite pure entangled states, so the network
is nonlocal by Refs. \cite{gisin_bells_1991,gisin_maximal_1992}.}. We label the edges as $k=1,...,|E|$ (where $|E|$ is the number
of edges of the graph) and the parties as $A_{1},...,A_{n}.$ Since
it will be enough to consider individual measurements on each particle,
we denote the input and output of party $A_{i}$ at edge $k$ as $x_{i}^{k},a_{i}^{k}$
respectively. We group the inputs and outputs of each party as $\chi_{i}=\{x_{i}^{k}\}_{k\in E_{i}},$
$\alpha_{i}=\{a_{i}^{k}\}_{k\in E_{i}}$ where $E_{i}$ is the set
of edges incident to vertex $i.$ Then, the shared distribution is
of the form 
\begin{equation}
P(\alpha_{1},...,\alpha_{n}|\chi_{1},...,\chi_{n})=\prod_{k=1}^{|E|}P_{k}(a_{i}^{k}a_{j}^{k}|x_{i}^{k}x_{j}^{k}),\label{eq:Pnpartite}
\end{equation}
where parties $A_{i},A_{j}$ are connected by edge $k$ (notice that
we label vertices and edges independently), and $P_{k}(a_{i}^{k}a_{j}^{k}|x_{i}^{k}x_{j}^{k})$
is the distribution arising from the state at edge $k.$ It will be
sufficient to consider tree graphs, i.e. such that every pair of vertices
(parties) is connected by exactly one path of edges. If the given
graph is not a tree, any extra edges can be ignored.

Depending on the nature of the shared states, we consider three cases: 
\begin{enumerate}
\item every shared state is less-than-maximally entangled;\label{enu:Hardy} 
\item every shared state is maximally entangled;\label{enu:maxent} 
\item some shared states are maximally entangled, some are not.\label{enu:mix} 
\end{enumerate}
\textbf{Case \ref{enu:Hardy}: }if all states are less-than-maximally
entangled, we prove the result by deriving an inequality that detects
GMNL and finding measurements on the shared states to violate it.
To derive the inequality, we will find bipartite inequalities that
can be violated by the state at each edge $k$, lift them to more
inputs, outputs and parties using the techniques in Ref. \cite{pironio_lifting_2005}
and combine them to obtain a GMNL inequality using tools in Ref. \cite{curchod_versatile_2019}.
We will consider 2-input 2-output measurements on each particle. Thus,
the global distribution will have $2^{|E_{i}|}$ inputs and outputs
for each party $A_{i}.$

We start from the inequality 
\begin{equation}
I=P(00|00)-P(01|01)-P(10|10)-P(00|11)\leq0,\label{eq:Iorig}
\end{equation}
which is a facet inequality equivalent to the CHSH inequality \cite{clauser_proposed_1969}
for nonsignalling distributions \cite{curchod_versatile_2019}. This
inequality detects any bipartite nonlocality present in any bipartition
that splits the parties connected by edge $k$ \cite{curchod_versatile_2019}.
To lift it to $n$ parties, each with $2^{|E_{i}|}$ inputs and outputs
(see Ref. \cite{pironio_lifting_2005}), we must set the inputs and
outputs of the parties that are not connected by edge $k$ to a fixed
value (0, wlog). For the parties $i$ that are connected by edge $k,$
any extra inputs other than $x_{i}^{k}=0_{i}^{k},1_{i}^{k}$ can be
ignored. Outputs must be grouped, by summing over some of their digits,
in order to get an effective 2-output distribution. It will be convenient
to add over the output components $a_{i}^{\bar{k}}$ that do not correspond
to edge $k,$ varying only the digit $a_{i}^{k}=0_{i}^{k},1_{i}^{k}.$
Thus, we obtain the following $n$-partite inequality at each edge
$k:$ 
\begin{equation}
\begin{aligned}I^{k}=\sum_{\overrightarrow{a}_{i}^{\bar{k}},\overrightarrow{a}_{j}^{\bar{k}}} & \left(P\left(\left.0_{i}^{k}\overrightarrow{a}_{i}^{\bar{k}},0_{j}^{k}\overrightarrow{a}_{j}^{\bar{k}},\overrightarrow{0}_{\bar{i},\bar{j}}\right|0_{i}^{k}0_{i}^{\bar{k}},0_{j}^{k}0_{j}^{\bar{k}},\overrightarrow{0}_{\bar{i},\bar{j}}\right)-P\left(\left.0_{i}^{k}\overrightarrow{a}_{i}^{\bar{k}},1_{j}^{k}\overrightarrow{a}_{j}^{\bar{k}},\overrightarrow{0}_{\bar{i},\bar{j}}\right|0_{i}^{k}0_{i}^{\bar{k}},1_{j}^{k}0_{j}^{\bar{k}},\overrightarrow{0}_{\bar{i},\bar{j}}\right)\right.\\
 & \left.-P\left(\left.1_{i}^{k}\overrightarrow{a}_{i}^{\bar{k}},0_{j}^{k}\overrightarrow{a}_{j}^{\bar{k}},\overrightarrow{0}_{\bar{i},\bar{j}}\right|1_{i}^{k}0_{i}^{\bar{k}},0_{j}^{k}0_{j}^{\bar{k}},\overrightarrow{0}_{\bar{i},\bar{j}}\right)-P\left(\left.0_{i}^{k}\overrightarrow{a}_{i}^{\bar{k}},0_{j}^{k}\overrightarrow{a}_{j}^{\bar{k}},\overrightarrow{0}_{\bar{i},\bar{j}}\right|1_{i}^{k}0_{i}^{\bar{k}},1_{j}^{k}0_{j}^{\bar{k}},\overrightarrow{0}_{\bar{i},\bar{j}}\right)\right)\leq0,
\end{aligned}
\label{eq:Ik}
\end{equation}
where the sum is over each binary digit $a_{i}^{\bar{k}},a_{j}^{\bar{k}}$
of the outputs of parties $i,j$ (which are connected by edge $k$),
except digits $a_{i}^{k},a_{j}^{k}$ which are fixed to 0 or 1 in
each term. The term $\overrightarrow{0}_{\bar{i},\bar{j}}$ denotes
input or output 0 for all components of all parties that are not $i,j.$
Thus, each inequality $I^{k}$ detects the bipartite nonlocality present
in the distribution $P$ across any bipartition that splits the parties
connected by edge $k.$ In the particular case of the distribution
(\ref{eq:Pnpartite}), it tells whether the component $P_{k}$ is
nonlocal.

Now, we can combine the inequalities $I^{k}$ to form a GMNL inequality:
\begin{equation}
I_{n}=\sum_{k=1}^{|E|}I^{k}+P(\overrightarrow{0},\overrightarrow{0}|\overrightarrow{0},\overrightarrow{0})-\sum_{k=1}^{|E|}\sum_{\overrightarrow{a}_{i}^{\bar{k}},\overrightarrow{a}_{j}^{\bar{k}}}P\left(\left.0_{i}^{k}\overrightarrow{a}_{i}^{\bar{k}},0_{j}^{k}\overrightarrow{a}_{j}^{\bar{k}},\overrightarrow{0}_{\bar{i},\bar{j}}\right|0_{i}^{k}0_{i}^{\bar{k}},0_{j}^{k}0_{j}^{\bar{k}},\overrightarrow{0}_{\bar{i},\bar{j}}\right)\leq0\,.\label{eq:InGMNL}
\end{equation}
To show that this is indeed a GMNL inequality, we must show that it
holds for any distribution $P$ that is local across some bipartition.
A bipartition of the network defines a cut of the graph. Because the
graph is assumed connected, for every cut there exists an edge $k_{0}$
which crosses the cut. Therefore, if $P$ is local across a bipartition
which is crossed by edge $k_{0},$ then by Ref. \cite{curchod_versatile_2019}
we have 
\begin{equation}
I^{k_{0}}\leq0.
\end{equation}
Hence, 
\begin{equation}
\begin{aligned}I_{n}\leq & \sum_{\substack{k=1\\
k\neq k_{0}
}
}^{|E|}I^{k}+P(\overrightarrow{0},\overrightarrow{0}|\overrightarrow{0},\overrightarrow{0})-\sum_{k=1}^{|E|}\sum_{\overrightarrow{a}_{i}^{\bar{k}},\overrightarrow{a}_{j}^{\bar{k}}}P\left(\left.0_{i}^{k}\overrightarrow{a}_{i}^{\bar{k}},0_{j}^{k}\overrightarrow{a}_{j}^{\bar{k}},\overrightarrow{0}_{\bar{i},\bar{j}}\right|0_{i}^{k}0_{i}^{\bar{k}},0_{j}^{k}0_{j}^{\bar{k}},\overrightarrow{0}_{\bar{i},\bar{j}}\right).\end{aligned}
\end{equation}

For each $k\neq k_{0},$ the only nonnegative term gets subtracted
in the final summation. The term $P(\overrightarrow{0},\overrightarrow{0}|\overrightarrow{0},\overrightarrow{0})$
then cancels out with the first term in the final summation for $k=k_{0},$
leaving only negative terms in the expression as required.

To complete the proof, we find local measurements for each party to
violate inequality (\ref{eq:InGMNL}). Since all shared states are
nonseparable and less-than-maximally entangled, the parties can choose
local measurements on each particle such that all resulting distributions
satisfy Hardy's paradox \cite{hardy_quantum_1992,hardy_nonlocality_1993}:
\begin{equation}
P_{k}(00|00)>0=P_{k}(01|01)=P_{k}(10|10)=P_{k}(00|11)\label{eq:Hardy}
\end{equation}
for each $k=1,...,|E|.$ This was proven for qubits in Refs. \cite{hardy_quantum_1992,hardy_nonlocality_1993},
and we show the extension to any local dimension in Proposition \ref{prop:d-dimHardy}
below. Because the distribution is of the form (\ref{eq:Pnpartite}),
each term in each inequality (\ref{eq:Ik}) simplifies significantly.
For example, the second term gives 
\begin{equation}
\begin{aligned}\sum_{\overrightarrow{a}_{i}^{\bar{k}},\overrightarrow{a}_{j}^{\bar{k}}} & P\left(\left.0_{i}^{k}\overrightarrow{a}_{i}^{\bar{k}},1_{j}^{k}\overrightarrow{a}_{j}^{\bar{k}},\overrightarrow{0}_{\bar{i},\bar{j}}\right|0_{i}^{k}0_{i}^{\bar{k}},1_{j}^{k}0_{j}^{\bar{k}},\overrightarrow{0}_{\bar{i},\bar{j}}\right)\\
 & =P_{k}(0_{i}^{k}1_{j}^{k}|0_{i}^{k}1_{j}^{k})\prod_{\ell}\sum_{a_{i}^{\ell}}P_{\ell}(a_{i}^{\ell}0_{j'}^{\ell}|0_{i}^{\ell}0_{j'}^{\ell})\prod_{\ell'}\sum_{a_{j}^{\ell'}}P_{\ell}(0_{i'}^{\ell'}a_{j}^{\ell'}|0_{i'}^{\ell'}0_{j}^{\ell'})\prod_{m}P_{m}(0_{i'}^{m}0_{j'}^{m}|0_{i'}^{m}0_{j'}^{m})\\
 & =P_{k}(0_{i}^{k}1_{j}^{k}|0_{i}^{k}1_{j}^{k})\,p_{k}\,,
\end{aligned}
\label{eq:summation1}
\end{equation}
where edges $\ell$ connect party $i$ to party $j'\neq j,$, edges
$\ell'$ connect party $j$ to party $i'\neq i,$ and edges $m$ connect
parties $i'$ and $j'$ where $i',j'\neq i,j.$ (Depending on the
structure of the graph, there may be no edges $\ell$, $\ell'$ or
$m$ for a given pair of parties $i,j,$ but that does not affect
the proof.)

The product of the terms $P_{\ell},\,P_{\ell'}$ and $P_{m}$ will
give a number $p_{k}.$ This is similar for the third and fourth terms,
which factorise to 
\begin{equation}
\begin{aligned}P_{k}(1_{i}^{k}0_{j}^{k}|1_{i}^{k}0_{j}^{k})\,p_{k},\\
P_{k}(0_{i}^{k}0_{j}^{k}|1_{i}^{k}1_{j}^{k})\,p_{k}
\end{aligned}
\label{eq:summation2}
\end{equation}
respectively. The first term of each $I^{k}$ cancels out with the
last summation in $I_{n},$ and the only term that remains is 
\begin{equation}
P(\overrightarrow{0},\overrightarrow{0}|\overrightarrow{0},\overrightarrow{0})=\prod_{k=1}^{|E|}P_{k}(0_{i}^{k}0_{j}^{k}|0_{i}^{k}0_{j}^{k}).
\end{equation}
Since $P_{k}$ satisfies Hardy's paradox for every $k$, then the
components of each $P_{k}$ appearing in equations (\ref{eq:summation1}),
(\ref{eq:summation2}) are all zero, while the only surviving term,
$P(\overrightarrow{0},\overrightarrow{0}|\overrightarrow{0},\overrightarrow{0}),$
is strictly greater than zero. Thus, the inequality $I_{n}$ is violated,
showing that $P$ is GMNL.

\textbf{Case \ref{enu:maxent}:} for every bipartition, there is an
edge that crosses the corresponding cut, and each of these edges already
contains a maximally entangled state. Therefore, the present network
meets the requirements of Theorem 2 in \cite{almeida_multipartite_2010},
so the network is GMNL---in fact it is fully GMNL.

\textbf{Case \ref{enu:mix}: }assume wlog that each edge $k=1,...,K$
contains a less-than-maximally entangled state, while each edge $k=K+1,...,|E|$
contains a maximally entangled state. Let 
\begin{equation}
P=P_{H}P_{+}\label{eq:PHP+}
\end{equation}
where 
\begin{equation}
\begin{aligned}P_{H}(\{a_{i}^{k}\}_{k\leq K,i\in[n]}|\{x_{i}^{k}\}_{k\leq K,i\in[n]}) & =\prod_{k=1}^{K}P_{k}(a_{i}^{k}a_{j}^{k}|x_{i}^{k}x_{j}^{k}),\\
P_{+}(\{a_{i}^{k}\}_{k>K,i\in[n]}|\{x_{i}^{k}\}_{k>K,i\in[n]}) & =\prod_{k=K+1}^{|E|}P_{k}(a_{i}^{k}a_{j}^{k}|x_{i}^{k}x_{j}^{k})
\end{aligned}
\end{equation}
where, on the right-hand side, parties $i,j$ are connected by edge
$k.$ For $k=1,...,K,$ terms $P_{k}$ satisfy Hardy's paradox (equation
(\ref{eq:Hardy})), as they arise from the measurements performed
in Case \ref{enu:Hardy}. For $k=K+1,...,|E|,$ the terms $P_{k}$
arise from measurements on the maximally entangled state to be specified
later. We now classify bipartitions depending on whether or not they
are crossed by an edge $k\leq K$ or $k>K:$ let $S_{\leq K}$ be
the set of bipartitions $M|\overline{M}$ (indexed by $M$) which
are crossed by an edge $k\leq K,$ and $T_{\leq K}$ be its complement,
i.e. the set of bipartitions which are \emph{not} crossed by an edge
$k\leq K.$ Similarly, $S_{>K}$ (respectively, $T_{>K}$) is the
set of bipartitions which are (not) crossed by an edge $k>K.$

Let $I_{H}^{k}$ be an inequality detecting nonlocality on edge $k,$
for the distribution $P_{H}.$ That is, $I_{H}^{k}$ is as in equation
(\ref{eq:Ik}) but where the sum over $\overrightarrow{a}_{i}^{\bar{k}},\overrightarrow{a}_{j}^{\bar{k}}$
concerns only the components of parties $A_{i},A_{j}$ that belong
only to edges $k'\leq K,k'\neq k.$ Then, consider the following functional
acting on distributions of the form of $P_{H}:$ 
\begin{equation}
I_{H}=\sum_{k=1}^{K}I_{H}^{k}+P(\overrightarrow{0},\overrightarrow{0}|\overrightarrow{0},\overrightarrow{0})-\sum_{k=1}^{K}\sum_{\overrightarrow{a}_{i}^{\bar{k}},\overrightarrow{a}_{j}^{\bar{k}}}P\left(\left.0_{i}^{k}\overrightarrow{a}_{i}^{\bar{k}},0_{j}^{k}\overrightarrow{a}_{j}^{\bar{k}},\overrightarrow{0}_{\bar{i},\bar{j}}\right|0_{i}^{k}0_{i}^{\bar{k}},0_{j}^{k}0_{j}^{\bar{k}},\overrightarrow{0}_{\bar{i},\bar{j}}\right).
\end{equation}
Again, the summation in the last term concerns only components that
belong to edges $k'\leq K,k'\neq k.$ We claim that the functional
$I_{H}$ is nonpositive for any distribution $P$ that is local across
a bipartition of type $S_{\leq K},$ i.e. one that is crossed by an
edge $k_{0}\leq K.$ The reasoning is similar to that in Case \ref{enu:Hardy}:
if $P$ is local across a bipartition crossed by an edge $k_{0}\leq K,$
then $I_{H}^{k_{0}}\leq0$ will be satisfied, and so 
\begin{equation}
I_{H}\leq\sum_{\substack{k=1\\
k\neq k_{0}
}
}^{K}I_{H}^{k}+P(\overrightarrow{0},\overrightarrow{0}|\overrightarrow{0},\overrightarrow{0})-\sum_{k=1}^{K}\sum_{\overrightarrow{a}_{i}^{\bar{k}},\overrightarrow{a}_{j}^{\bar{k}}}P\left(\left.0_{i}^{k}\overrightarrow{a}_{i}^{\bar{k}},0_{j}^{k}\overrightarrow{a}_{j}^{\bar{k}},\overrightarrow{0}_{\bar{i},\bar{j}}\right|0_{i}^{k}0_{i}^{\bar{k}},0_{j}^{k}0_{j}^{\bar{k}},\overrightarrow{0}_{\bar{i},\bar{j}}\right).
\end{equation}
Now, for each $k\neq k_{0},$ the only nonnegative term gets subtracted
in the final summation. The term $P(\overrightarrow{0},\overrightarrow{0}|\overrightarrow{0},\overrightarrow{0})$
then cancels out with the first term in the final summation for $k=k_{0},$
leaving only negative terms in the expression as required.

We now show that, for $P=P_{H},$ we have $I_{H}>0.$ Indeed, the
terms in $I_{H}^{k}$ simplify in a similar manner to Case \ref{enu:Hardy}.
Then, since each $P_{k},$ $k\leq K$ satisfies Hardy's paradox, the
second, third and fourth terms in each $I_{H}^{k}$ are zero, the
first cancels out with the last summation, and the only surviving
term is 
\begin{equation}
P(\overrightarrow{0},\overrightarrow{0}|\overrightarrow{0},\overrightarrow{0})=\prod_{k=1}^{K}P_{k}(0_{i}^{k}0_{j}^{k}|0_{i}^{k}0_{j}^{k})>0.
\end{equation}

This means that there exists an $\varepsilon>0$ such that, for any
EPR2 decomposition of $P_{H},$ 
\begin{equation}
P_{H}=\sum_{M}p_{L,H}^{M}P_{L,H}^{M}+p_{NS,H}P_{NS,H},
\end{equation}
we have that the terms where $P_{L,H}^{M}$ is local across a bipartition
such that $M\in S_{\leq K}$ satisfy 
\begin{equation}
\sum_{M\in S_{\le K}}p_{L,H}^{M}\leq1-\varepsilon.\label{eq:plocal-Hardy}
\end{equation}

Also, it can be deduced from Ref. \cite{almeida_multipartite_2010}
that, given the $\varepsilon$ above, the parties can choose suitable
measurements such that $P_{+}$ is fully nonlocal across all bipartitions
$S_{>K}.$ That is, any multipartite EPR2 decomposition of $P_{+},$
\begin{equation}
P_{+}=\sum_{M}p_{L,+}^{M}P_{L,+}^{M}+p_{NS,+}P_{NS,+},
\end{equation}
is such that the terms where $P_{L,+}^{M}$ is local across a bipartition
such that $M\in S_{>K}$ satisfy 
\begin{equation}
\sum_{M\in S_{>K}}p_{L,+}^{M}<\varepsilon.\label{eq:plocal-maxent}
\end{equation}

To prove that the global distribution $P$ is GMNL, as is our goal,
we assume the converse, and we derive a contradiction from the nonlocality
properties of $P_{H}$ and $P_{+}.$ Assuming $P$ is not GMNL, we
can express the distribution as 
\begin{equation}
P=\sum_{\lambda,M}p_{L}^{M}(\lambda)P_{M}(\{\alpha_{i}\}_{i\in M}|\{\chi_{i}\}_{i\in M},\lambda)P_{\overline{M}}(\{\alpha_{i}\}_{i\in\overline{M}}|\{\chi_{i}\}_{i\in\overline{M}},\lambda),\label{eq:PnotGMNL}
\end{equation}
where $p_{L}^{M}(\lambda)$ are nonegative numbers for every $M,\lambda$
such that 
\begin{equation}
\sum_{\lambda,M}p_{L}^{M}(\lambda)=1,
\end{equation}
for each $\alpha_{i},\chi_{i},i=1,...,n.$

Then, summing over the output components $a_{i}^{k}$ for all $k\leq K$
and all $i,$ we get $P_{+}$ on the left-hand side, from equation
(\ref{eq:PHP+}). On the right-hand side, we get two types of terms
(depending on the type of bipartition) that turn out to form an EPR2
decomposition of $P_{+}$ \footnote{Note that, while each of the terms on the right-hand side may depend
on the whole of each party's input $\chi_{i},$ the left-hand side
does not, because the distribution is of the form (\ref{eq:PHP+}).
That is, the resulting EPR2 decomposition of $P_{+}$ holds for any
fixed value of the inputs $\{x_{i}^{k}\}_{k\leq K}$ on the left-
and right-hand sides.}. Indeed, the local terms are given by bipartitions such that $M\in S_{>K},$
while the nonlocal terms are given by bipartitions such that $M\in T_{>K}$
(since all terms are nonsignalling). By equation (\ref{eq:plocal-maxent}),
the choice of measurements on the particles involved in $P_{+}$ ensures
that 
\begin{equation}
\sum_{\lambda,M\in S_{>K}}p_{L}^{M}(\lambda)<\varepsilon,
\end{equation}
while 
\begin{equation}
\sum_{\lambda,M\in T_{>K}}p_{L}^{M}(\lambda)>1-\varepsilon.\label{eq:sumI1-1}
\end{equation}

If, instead, we sum over the output components $a_{i}^{k}$ for all
$k>K$ and all $i,$ we get $P_{H}$ on the left-hand side, from equation
(\ref{eq:PHP+}). On the right-hand side, by similar reasoning we
find an EPR2 decomposition of $P_{H}.$ This time, $S_{\leq K}$ will
give the local terms and $T_{\leq K}$ will give the nonlocal terms.
By equation (\ref{eq:plocal-Hardy}), we have 
\begin{equation}
\sum_{\lambda,M\in S_{\leq K}}p_{L}^{M}(\lambda)\leq1-\varepsilon.\label{eq:sumJ2}
\end{equation}
Now, since the graph is connected, if a bipartition is not crossed
by an edge $k>K,$ then it must be crossed by an edge $k\leq K.$
That is, $T_{>K}\subseteq S_{\leq K}.$ This means that equation (\ref{eq:sumJ2})
also holds if the sum is over $T_{>K},$ but this contradicts equation
(\ref{eq:sumI1-1}). Therefore, the distribution $P$ must be GMNL.
\end{proof}
In Theorem \ref{thm:gmnl-from-bipartite-ent} we assumed that all
less-than-maximally entangled states satisfy Hardy's paradox. This
is shown for qubits in \cite{hardy_nonlocality_1993}, and we now
extend the proof to any dimension.

\setcounter{thm}{0} 
\begin{prop}
\label{prop:d-dimHardy}Let $\ket{\psi}\in\mathcal{H}_{A}\otimes\mathcal{H}_{B}\cong\left(\mathbb{C}^{d}\right)^{\otimes2}$
be a nonseparable and less-than-maximally entangled pure state . Then,
$\ket{\psi}$ satisfies Hardy's paradox. 
\end{prop}

\begin{proof}
Let $\ket{\psi}$ be as in the statement of the Proposition. We present
2-input, 2-output measurements for $\ket{\psi}$ to generate a distribution
which satisfies Hardy's paradox \cite{hardy_quantum_1992,hardy_nonlocality_1993}
using tools from Ref. \cite{curchod_versatile_2019}.

Consider the Schmidt decomposition 
\begin{equation}
\ket{\psi}=\sum_{i=0}^{d-1}\lambda_{i}^{1/2}\ket{ii}
\end{equation}
and assume the coefficients are ordered such that $0\neq\lambda_{0}\neq\lambda_{1}\neq0,$
which is always possible if the state is nonseparable and less-than-maximally
entangled. Wlog assume the Schmidt basis of the state is the canonical
basis. Let $\alpha\in]0,\pi/2[$ and $\delta\in\mathbb{R}$ and consider
the dual vectors 
\begin{equation}
\begin{aligned}\bra{e_{0|0}} & =\cos\alpha\bra{0}+\ue^{\ui\delta}\sin\alpha\bra{1}\\
\bra{e_{1|1}} & =\lambda_{0}\cos\alpha\bra{0}+\lambda_{1}\ue^{\ui\delta}\sin\alpha\bra{1}\\
\bra{f_{0|0}} & =\lambda_{1}^{3/2}\ue^{\ui\delta}\sin\alpha\bra{0}-\lambda_{0}^{3/2}\cos\alpha\bra{1}\\
\bra{f_{1|1}} & =\lambda_{1}^{1/2}\ue^{\ui\delta}\sin\alpha\bra{0}-\lambda_{0}^{1/2}\cos\alpha\bra{1}
\end{aligned}
\label{eq:Hardymts}
\end{equation}
(one can write the projectors in the Schmidt basis of the state instead
of assuming the state decomposes into the canonical basis). Define
the measurements $E_{a|x}$ for Alice, with input $x$ and output
$a,$ and $F_{b|y}$ for Bob, with input $y$ and output $b,$ given
by 
\begin{equation}
\begin{aligned}E_{0|0} & =\ketbra{e_{0|0}}\\
E_{1|0} & \propto\ketbra{e_{0|0}}^{\perp}\oplus\mathbbm1_{2,...,d-1}\\
E_{0|1} & \propto\ketbra{e_{1|1}}^{\perp}\\
E_{1|1} & \propto\ketbra{e_{1|1}}\oplus\mathbbm1_{2,...,d-1}\\
F_{0|0} & \propto\ketbra{f_{0|0}}\\
F_{1|0} & \propto\ketbra{f_{0|0}}^{\perp}\oplus\mathbbm1_{2,...,d-1}\\
F_{0|1} & \propto\ketbra{f_{1|1}}^{\perp}\oplus\mathbbm1_{2,...,d-1}\\
F_{1|1} & \propto\ketbra{f_{1|1}}
\end{aligned}
\end{equation}
where $\ketbra{e_{0|0}}^{\perp}$ denotes the density matrix corresponding
to the vector orthogonal to $\ket{e_{0|0}}$ when restricted to the
subspace spanned by $\{\ket{0},\ket{1}\},$ and $\mathbbm1_{2,...,d-1}$
is the identity operator on the subspace spanned by $\left\{ \ket{i}\right\} _{i=2}^{d-1},$
for either Alice or Bob. Note that, since we are only interested in
whether some probabilities are equal or different from zero, normalisation
will not play a role.

We now show that the distribution given by 
\begin{equation}
P(ab|xy)=\tr(E_{a|x}\otimes F_{b|y}\ketbra{\psi})
\end{equation}
satisfies Hardy's paradox. Indeed, because of the probabilities considered
and the form of the measurements, only the terms in $i=0,1$ contribute
to the probabilities that appear in Hardy's paradox, therefore

\begin{equation}
\begin{aligned}P(01|01)\propto\left|\sum_{i=0}^{1}\lambda_{i}^{1/2}\left(\bra{e_{0|0}}\otimes\bra{f_{1|1}}\right)\ket{ii}\right|^{2} & =0\\
P(10|10)\propto\left|\sum_{i=0}^{1}\lambda_{i}^{1/2}\left(\bra{e_{1|1}}\otimes\bra{f_{0|0}}\right)\ket{ii}\right|^{2} & =0\\
P(00|11)\propto\left|\sum_{i=0}^{1}\lambda_{i}^{1/2}\left(\bra{e_{0|1}}\otimes\bra{f_{0|1}}\right)\ket{ii}\right|^{2} & =0.
\end{aligned}
\end{equation}
For $P(00|00),$ we find 
\begin{equation}
\begin{aligned}P(00|00) & \propto\left|\sum_{i=0}^{1}\lambda_{i}^{1/2}\left(\bra{e_{0|0}}\otimes\bra{f_{0|0}}\right)\ket{ii}\right|^{2}\\
 & =\left|\ue^{\ui\delta}\sin\alpha\cos\alpha\,\lambda_{0}^{1/2}\,\lambda_{1}^{1/2}\,(\lambda_{1}-\lambda_{0})\right|^{2},
\end{aligned}
\end{equation}
which is strictly greater than zero when $\alpha\in]0,\pi/2[$ and
$0\neq\lambda_{0}\neq\lambda_{1}\neq0,$ like we assumed. This proves
the claim. 
\end{proof}

\setcounter{thm}{1} 
\paragraph{GMNL from GME\protect \protect \\
 }

We fix some notation that we will use in Theorem \ref{thm:copies}.
The result considers a GME state $\ket{\Psi}\in\mathcal{H}_{A}\otimes\mathcal{H}_{B_{1}}...\otimes\mathcal{H}_{B_{n-1}}\cong(\mathbb{C}^{d})^{\otimes n},$
$n-1$ copies of which are shared between $n$ parties $A,B_{1},...,B_{n-1}.$
Each party measures locally on each particle, like in Theorem \ref{thm:gmnl-from-bipartite-ent}.
We denote Alice's input and output, respectively, as $\chi\equiv x_{1}...x_{n-1},$
$\alpha\equiv a_{1}...a_{n-1}$ in terms of the digits $x_{i},a_{i}$
corresponding to each particle $i\in[n-1].$ We let the measurement
made by party $B_{j}$ on copy $i$ have input $y_{j}^{i}$ and output
$b_{j}^{i}\,,$ where $i,j=1,...,n-1,$ and for each $j$ we denote
$\upsilon_{j}=y_{j}^{1}...y_{j}^{n-1}$ and $\beta_{j}=b_{j}^{1}...b_{j}^{n-1}$
digit-wise. Then, after measurement, the parties share a distribution
\begin{equation}
\left\{ P(\alpha\beta_{1}...\beta_{n-1}|\chi\upsilon_{1}...\upsilon_{n-1})\right\} _{\substack{\alpha,\beta_{1}...\beta_{n-1}\\
\chi,\upsilon_{1}...\upsilon_{n-1}
}
}.
\end{equation}
Because we are considering local measurements made on each particle,
this distribution is of the form 
\begin{equation}
P(\alpha\beta_{1}...\beta_{n-1}|\chi\upsilon_{1}...\upsilon_{n-1})=\prod_{i=1}^{n-1}P_{i}(a_{i}b_{1}^{i}...b_{n-1}^{i}|x_{i}y_{1}^{i}...y_{n-1}^{i}),\label{eq:Pnpartite-fromGME}
\end{equation}
where each $P_{i}$ is the distribution arising from copy $i$ of
the state $\ket{\Psi}.$ As advanced in the main text, each copy $i$
of the state $\ket{\Psi}$ will give an edge of a star network connecting
Alice and party $B_{i}.$ Because of the structure of this particular
network, we can simplify the notation with respect to Theorem \ref{thm:gmnl-from-bipartite-ent}
and identify the index of each party $B_{i}$ with its corresponding
edge $i.$
\begin{thm}
\label{thm:copies}Any GME state $\ket{\Psi}\in\mathcal{H}_{A}\otimes\mathcal{H}_{B_{1}}...\otimes\mathcal{H}_{B_{n-1}}\cong(\mathbb{C}^{d})^{\otimes n}$
is such that $\ket{\Psi}^{\otimes(n-1)}$ is GMNL. 
\end{thm}

\begin{proof}
For each copy $i=1,...,n-1$ of the state $\ket{\Psi},$ we will find
measurements for parties $\{B_{j}\}_{j\neq i}$ that leave Alice and
party $B_{i}$ with a bipartite entangled state. This will yield a
network in a similar configuration to Theorem \ref{thm:gmnl-from-bipartite-ent}
for a star network, but conditionalised on the inputs and outputs
of these measurements. We will generalise the result of Theorem \ref{thm:gmnl-from-bipartite-ent}
as it applies to a star network to show that this network is also
GMNL.

Let $i\in[n-1]$ and consider the $i$th copy of $\ket{\Psi}.$ Suppose
each party $B_{j},$ $j\neq i,$ performs a local, projective measurement
onto a basis $\left\{ \ket{b_{j}}\right\} _{b_{j}=0}^{d-1}.$ We pick
the computational basis on each party's Hilbert space to be such that
the measurement performed by the parties $B_{j},$ $j\neq i,$ leave
Alice and $B_{i}$ in state $|\phi_{\overrightarrow{b}}\rangle_{AB_{i}},$
where $\overrightarrow{b}=b_{1}...b_{i-1}b_{i+1}...b_{n-1}$ denotes
the output obtained by the parties $B_{j},$ $j\neq i$ (we briefly
omit the script $i$ referring to the copy of the state, for readability).
This means that we can write the state $\ket{\Psi}$ as 
\begin{equation}
\ket{\Psi}=\sum_{\overrightarrow{b}}\lambda_{\overrightarrow{b}}\ket{\phi_{\overrightarrow{b}}}_{AB_{i}}\ket{\overrightarrow{b}}_{B_{1}...B_{i-1}B_{i+1}...B_{n-1}.}
\end{equation}
Ref. \cite{popescu_generic_1992}, whose proof was completed in Ref.
\cite{gachechiladze_completing_2017}, showed that there always exist
measurements (i.e. bases) $\left\{ \ket{b_{j}}\right\} _{b_{j}=0}^{d-1}$
such that $\ket{\phi_{\overrightarrow{b}}}_{AB_{i}}$ is entangled
for a certain output $\overrightarrow{b}$. We now show that this
opens up only two possibilities for each $i$: either there exists
an output such that $\ket{\phi_{\overrightarrow{b}}}_{AB_{i}}$ is
less-than-maximally entangled, or for all outputs $\overrightarrow{b}$,
$\ket{\phi_{\overrightarrow{b}}}_{AB_{i}}$ is maximally entangled.
Indeed, the only option left to discard is one where, for some $\overrightarrow{b}=\overrightarrow{b^{*}},$
$\ket{\phi_{\overrightarrow{b^{*}}}}_{AB_{i}}$ is maximally entangled,
and for some other $\overrightarrow{b}=\overrightarrow{b^{**}},$
$\ket{\phi_{\overrightarrow{b^{**}}}}_{AB_{i}}$ is separable. But
it is easy to see, by using a continuity argument, that in this case
the bases $\left\{ \ket{b_{j}}\right\} _{b_{j}=0}^{d-1}$ can be modified
so that there exists one output for which $AB_{i}$ are projected
onto a less-than-maximally entangled state: it suffices to consider
one (normalised) element of the measurement basis to be $c_{0}\ket{b_{j}^{*}}+c_{1}\ket{b_{j}^{**}}$
for some values $c_{0},c_{1}\in\mathbb{C},$ for each $j.$

Therefore, we consider the following cases: 
\begin{enumerate}
\item for all $i\in[n-1],$ there exists an input and output for each $B_{j},j\neq i$
such that $\ket{\phi_{i}}_{AB_{i}}$ is less-than-maximally entangled;\label{enu:Hardy-copies} 
\item for all $i\in[n-1],$ there exists an input for each $B_{j},j\neq i$
such that $\ket{\phi_{i}}_{AB_{i}}$ is maximally entangled for all
outputs;\label{enu:maxent-copies} 
\item there exist $i,k\in[n-1]$ such that $\ket{\phi_{i}}_{AB_{i}}$ is
as in Case \ref{enu:maxent-copies} and $\ket{\phi_{k}}_{AB_{k}}$
is as in Case \ref{enu:Hardy-copies}.\label{enu:mix-copies} 
\end{enumerate}
\textbf{Case \ref{enu:Hardy-copies}:} let $i\in[n-1].$ Suppose parties
$\{B_{j}\}_{j\neq i}$ perform the measurements explained above that
leave Alice and $B_{i}$ less-than-maximally entangled. Then, Alice
and $B_{i}$ can perform local measurements on the resulting state
to satisfy Hardy's paradox. We will modify the inequality in Theorem
\ref{thm:gmnl-from-bipartite-ent} and show that these measurements
on $\ket{\Psi}^{\otimes(n-1)}$ give a distribution which violates
the inequality.

To modify the inequality in Theorem \ref{thm:gmnl-from-bipartite-ent},
we import the same strategy to lift inequality (\ref{eq:Iorig}) to
$n$ parties, each with $2^{n-1}$ inputs and outputs. We want $I^{AB_{i}}$
to detect bipartite nonlocality between Alice's $i$th particle and
$B_{i}$'s $i$th particle, that is, nonlocality in $a_{i}b_{i}^{i}|x_{i}y_{i}^{i}.$
Therefore, for each $i$ we now need to fix all other inputs $x_{j},y_{i}^{j},y_{j}^{j}$
and add over all other outputs $a_{j},b_{i}^{j},b_{j}^{j},$ $j\neq i,$
so that 
\begin{equation}
\begin{aligned}I^{AB_{i}}=\sum_{a_{\bar{i}},b_{i}^{\bar{i}},b_{\bar{i}}^{\bar{i}}=0,1} & \left(P(0_{i}a_{\bar{i}}\,,0_{i}^{i}b_{i}^{\bar{i}}\,,0_{\bar{i}}^{i}b_{\bar{i}}^{\bar{i}}\,|0_{i}0_{\bar{i}}\,,0_{i}^{i}0_{i}^{\bar{i}}\,,0_{\bar{i}}^{i}0_{\bar{i}}^{\bar{i}})-P(0_{i}a_{\bar{i}}\,,1_{i}^{i}b_{i}^{\bar{i}}\,,0_{\bar{i}}^{i}b_{\bar{i}}^{\bar{i}}\,|0_{i}0_{\bar{i}}\,,1_{i}^{i}0_{i}^{\bar{i}}\,,0_{\bar{i}}^{i}0_{\bar{i}}^{\bar{i}})\right.\\
 & \left.-P(1_{i}a_{\bar{i}}\,,0_{i}^{i}b_{i}^{\bar{i}}\,,0_{\bar{i}}^{i}b_{\bar{i}}^{\bar{i}}\,|1_{i}0_{\bar{i}}\,,0_{i}^{i}0_{i}^{\bar{i}}\,,0_{\bar{i}}^{i}0_{\bar{i}}^{\bar{i}})-P(0_{i}a_{\bar{i}}\,,0_{i}^{i}b_{i}^{\bar{i}}\,,0_{\bar{i}}^{i}b_{\bar{i}}^{\bar{i}}\,|1_{i}0_{\bar{i}}\,,1_{i}^{i}0_{i}^{\bar{i}}\,,0_{\bar{i}}^{i}0_{\bar{i}}^{\bar{i}})\right)\,,
\end{aligned}
\end{equation}
where the outputs in the first term are denoted as follows: $0_{i}a_{\bar{i}}$
denotes output $\alpha=a_{1}...0_{i}...a_{n-1},\:$ $0_{i}^{i}b_{i}^{\bar{i}}$
denotes output $\beta_{i}=b_{i}^{1}...0_{i}^{i}...b_{i}^{n-1},\:$
and $0_{\bar{i}}^{i}b_{\bar{i}}^{\bar{i}}$ denotes output $\beta_{j}=b_{j}^{1}...0_{j}^{i}...b_{j}^{n-1}$
for \emph{all} $j\neq i.$ Inputs are denoted similarly, and the notation
is similar for the other three terms. Then, the inequality 
\begin{equation}
I_{n}=\sum_{i=1}^{n-1}I^{AB_{i}}+P(\overrightarrow{0},\overrightarrow{0}|\overrightarrow{0},\overrightarrow{0})-\sum_{i=1}^{n-1}\sum_{a_{\bar{i}},b_{i}^{\bar{i}},b_{\bar{i}}^{\bar{i}}=0,1}P(0_{i}a_{\bar{i}}\,,0_{i}^{i}b_{i}^{\bar{i}}\,,0_{\bar{i}}^{i}b_{\bar{i}}^{\bar{i}}\,|0_{i}0_{\bar{i}}\,,0_{i}^{i}0_{i}^{\bar{i}}\,,0_{\bar{i}}^{i}0_{\bar{i}}^{\bar{i}})\leq0
\end{equation}
is a GMNL inequality, by the same reasoning as in Theorem \ref{thm:gmnl-from-bipartite-ent}.

Evaluating the inequality on the distribution (\ref{eq:Pnpartite-fromGME}),
we find again that each term simplifies. For each $i$ we get, for
example, 

\begin{equation}
\begin{aligned}\sum_{a_{\bar{i}},b_{i}^{\bar{i}},b_{\bar{i}}^{\bar{i}}=0,1}P & (0_{i}a_{\bar{i}}\,,1_{i}^{i}b_{i}^{\bar{i}}\,,0_{\bar{i}}^{i}b_{\bar{i}}^{\bar{i}}\,|0_{i}0_{\bar{i}}\,,1_{i}^{i}0_{i}^{\bar{i}}\,,0_{\bar{i}}^{i}0_{\bar{i}}^{\bar{i}})\\
 & =P_{i}(0_{i}1_{i}^{i}0_{\bar{i}}^{i}|0_{i}1_{i}^{i}0_{\bar{i}}^{i})\prod_{\substack{j=1\\
j\neq i
}
}^{n-1}\sum_{\substack{a_{j},b_{k}^{j}=0,1\\
k\neq j
}
}P_{j}(a_{j}b_{1}^{j}...b_{j-1}^{j}b_{j+1}^{j}...b_{n-1}^{j}\,|0_{j}0_{1}^{j}...0_{j-1}^{j}0_{j+1}^{j}...0_{n-1}^{j})\\
 & =P_{i}(0_{i}1_{i}^{i}0_{\bar{i}}^{i}|0_{i}1_{i}^{i}0_{\bar{i}}^{i})
\end{aligned}
\end{equation}
and, similarly, 
\begin{equation}
\begin{aligned}\sum_{a_{\bar{i}},b_{i}^{\bar{i}},b_{\bar{i}}^{\bar{i}}=0,1}P(1_{i}a_{\bar{i}}\,,0_{i}^{i}b_{i}^{\bar{i}}\,,0_{\bar{i}}^{i}b_{\bar{i}}^{\bar{i}}\,|1_{i}0_{\bar{i}}\,,0_{i}^{i}0_{i}^{\bar{i}}\,,0_{\bar{i}}^{i}0_{\bar{i}}^{\bar{i}}) & =P_{i}(1_{i}0_{i}^{i}0_{\bar{i}}^{i}|1_{i}0_{i}^{i}0_{\bar{i}}^{i})\,;\\
\sum_{a_{\bar{i}},b_{i}^{\bar{i}},b_{\bar{i}}^{\bar{i}}=0,1}P(0_{i}a_{\bar{i}}\,,0_{i}^{i}b_{i}^{\bar{i}}\,,0_{\bar{i}}^{i}b_{\bar{i}}^{\bar{i}}\,|1_{i}0_{\bar{i}}\,,1_{i}^{i}0_{i}^{\bar{i}}\,,0_{\bar{i}}^{i}0_{\bar{i}}^{\bar{i}}) & =P_{i}(0_{i}0_{i}^{i}0_{\bar{i}}^{i}|1_{i}1_{i}^{i}0_{\bar{i}}^{i}).
\end{aligned}
\end{equation}
Also, 
\begin{equation}
P(\overrightarrow{0},\overrightarrow{0}|\overrightarrow{0},\overrightarrow{0})=\prod_{i=1}^{n-1}P_{i}(0_{i}0_{i}^{i}0_{\bar{i}}^{i}|0_{i}0_{i}^{i}0_{\bar{i}}^{i})\,.
\end{equation}
Now each $P_{i}$ in equation (\ref{eq:Pnpartite-fromGME}) arises
from measurements by $\{B_{j}\}_{j\neq i}$ to create a less-than-maximally
entangled state between Alice and $B_{i},$ who can then choose measurements
to satisfy Hardy's paradox. Hence all terms are zero except $P(\overrightarrow{0},\overrightarrow{0}|\overrightarrow{0},\overrightarrow{0})>0,$
and so the inequality is violated. Therefore, $\ket{\Psi}^{\otimes(n-1)}$
is GMNL.

\textbf{Case \ref{enu:maxent-copies}:} we assumed that, for all $i\in[n-1],$
there exist local measurements on $\ket{\Psi}$ for parties $\left\{ B_{j}\right\} _{j\neq i}$
that, for \emph{all} outcomes, create a maximally entangled state
$\ket{\phi_{i}}_{AB_{i}}$ shared between Alice and $B_{i}.$ Since
all bipartitions can be expressed as $A|B_{i}$ for some $i,$ we
find that $\ket{\Psi}$ meets the requirements of Theorem 2 in \cite{almeida_multipartite_2010},
and so $\ket{\Psi}$ is GMNL. That is, one copy of the shared state
$\ket{\Psi}$ is already GMNL, and therefore so is $\ket{\Psi}^{\otimes(n-1)}.$

\textbf{Case \ref{enu:mix-copies}:} assume wlog that the state $\ket{\phi_{i}}_{AB_{i}}$
is less-than-maximally entangled for $i=1,...,K$ and maximally entangled
for $i=K+1,...,n-1.$ We will show that $\ket{\Psi}^{\otimes(K+1)}$
is GMNL, which implies that $\ket{\Psi}^{\otimes(n-1)}$ is so too.

It will be useful to classify bipartitions $M|\overline{M}$ like
in Theorem \ref{thm:gmnl-from-bipartite-ent}. We will always assume
that Alice belongs to $M$ in order not to duplicate the bipartitions.
Let $S_{\leq K}$ be the set of bipartitions $M|\overline{M}$ (indexed
by $M$) which are crossed by an edge $j\leq K,$ i.e., where $\overline{M}$
contains at least one index $j\in\{1,...,K\},$ and $T_{\leq K}$
be its complement, i.e. the set of bipartitions where $\overline{M}$
contains only indices $j\in\{K+1,...,n-1\}.$ Similarly, $S_{>K}$
(respectively, $T_{>K}$) is the set of bipartitions which are (not)
crossed by an edge $j>K.$ That is, in $S_{>K},$ there is some $j\in\{K+1,...,n-1\}$
which belongs to $\overline{M},$ while in $T_{>K},$ $\overline{M}$
contains only indices $j\in\{1,...,K\}.$

For each $i=1,...,K,$ parties $AB_{i}$ can perform measurements
on their shared state $\ket{\phi_{i}}_{AB_{i}}$ which, together with
the measurements of parties $\{B_{j}\}_{j\neq i}$ that projected
$\ket{\Psi}$ onto $\ket{\phi_{i}}_{AB_{i}},$ give rise to a distribution
\begin{equation}
P_{i}(a_{i}b_{1}^{i}...b_{n-1}^{i}|x_{i}y_{1}^{i}...y_{n-1}^{i})\label{eq:PiHardy-copies}
\end{equation}
which satisfies Hardy's paradox when post-selected on the inputs and
outputs of parties $\{B_{j}\}_{j\neq i}.$ Then, the distribution
arising from the first $K$ copies of $\ket{\Psi}$ is 
\begin{equation}
P_{H}(\{a_{i}\}_{i\leq K}\{b_{j}^{i}\}_{i\leq K,j\in[n-1]}|\{x_{i}\}_{i\leq K}\{y_{j}^{i}\}_{i\leq K,j\in[n-1]})=\prod_{i=1}^{K}P_{i}(a_{i}b_{1}^{i}...b_{n-1}^{i}|x_{i}y_{1}^{i}...y_{n-1}^{i}),\label{eq:PHardy-copies}
\end{equation}
with $P_{i}$ as in equation (\ref{eq:PiHardy-copies}). This distribution
is similar to that in Case \ref{enu:Hardy-copies} when post-selected
on the inputs and outputs of parties $\{B_{j}\}_{j>K}.$ More precisely,
by the nonsignalling condition, we have 
\begin{equation}
\begin{aligned}P_{H}(\{a_{i}\}_{i\leq K} & \{b_{j}^{i}\}_{i\leq K,j\leq K}\{b_{j}^{i}=0_{j}^{i}\}_{i\leq K,j>K}|\{x_{i}\}_{i\leq K}\{y_{j}^{i}\}_{i\leq K,j\leq K}\{y_{j}^{i}=0_{j}^{i}\}_{i\leq K,j>K})=\\
 & P_{AB_{1}...B_{K}}(\{a_{i}\}_{i\leq K}\{b_{j}^{i}\}_{i\leq K,j\leq K}|\{x_{i}\}_{i\leq K}\{y_{j}^{i}\}_{i\leq K,j\leq K},\{b_{j}^{i}=0_{j}^{i}\}_{i\leq K,j>K},\{y_{j}^{i}=0_{j}^{i}\}_{i\leq K,j>K})\\
 & \times P_{B_{K+1}...B_{n-1}}(\{b_{j}^{i}=0_{j}^{i}\}_{i\leq K,j>K}|\{y_{j}^{i}=0_{j}^{i}\}_{i\leq K,j>K}),
\end{aligned}
\label{eq:PHcondtl-copies}
\end{equation}
where by Case \ref{enu:Hardy-copies} we know that $P_{AB_{1}...B_{K}}$
is GMNL in its parties. Then, $P_{H}$ must be $(K+1)$-way nonlocal
(i.e., GMNL when restricted to parties $A,B_{1},...,B_{K}$). Indeed,
if this were not the case, by equation (\ref{eq:PHcondtl-copies})
we could obtain a decomposition of the form (\ref{eq:BLdistrNpartite})
for $P_{AB_{1}...B_{K}},$ which would contradict the fact that this
distribution is GMNL.

Therefore, there exists an $\varepsilon>0$ such that any EPR2 decomposition
of $P_{H}$ as 
\begin{equation}
P_{H}=\sum_{M}p_{L,H}^{M}P_{L,H}^{M}+p_{NS,H}P_{NS,H}
\end{equation}
we have that the terms where $P_{L,H}^{M}$ is local across a bipartition
such that $M\in S_{\leq K}$ satisfy 
\begin{equation}
\sum_{M\in S_{\leq K}}p_{L,H}^{M}\leq1-\varepsilon.\label{eq:pHlocal-copies}
\end{equation}

On the other hand, $\ket{\Psi}$ satisfies Theorem 1 in Ref. \cite{almeida_multipartite_2010}
for all bipartitions $A|B_{i}$ for $i=K+1,...,n-1,$ hence it is
fully nonlocal across all such bipartitions. This means that, for
any $\delta_{i}>0,$ there exist local measurements on $\ket{\Psi}$
(which depend on $i$) that lead to a distribution 
\begin{equation}
P_{+}(ab_{1}...b_{n-1}|xy_{1}...y_{n-1})\label{eq:P+-copies}
\end{equation}
such that any bipartite EPR2 decomposition across a bipartition $A|B_{i},$
for $i=K+1,...,n-1,$ 
\begin{equation}
P_{+}=p_{L,+}^{A|B_{i}}P_{L,+}^{A|B_{i}}+(1-p_{L,+}^{A|B_{i}})P_{NS,+}^{A|B_{i}}
\end{equation}
satisfies 
\begin{equation}
p_{L,+}^{A|B_{i}}<\delta_{i}.\label{eq:P+localdeltai-copies}
\end{equation}
Thus, considering the possibility of implementing all the above measurements
for each \emph{i} leads to a distribution of the form (\ref{eq:P+-copies})
in which equation (\ref{eq:P+localdeltai-copies}) holds for every
$i=K+1,...,n-1.$

Therefore, given the $\varepsilon$ above, the parties can choose
suitable $\delta_{i}$ to bound the bipartitely local components and
hence ensure that any multipartite EPR2 decomposition of $P_{+},$
\begin{equation}
P_{+}=\sum_{M}p_{L,+}^{M}P_{L,+}^{M}+p_{NS,+}P_{NS,+}
\end{equation}
is such that the terms where $P_{L,+}^{M}$ is local across a bipartition
such that $M\in S_{>K}$ satisfy 
\begin{equation}
\sum_{M\in S_{>K}}p_{L,+}^{M}<\varepsilon.\label{eq:Plocal-maxent-copies}
\end{equation}

Since we only need to consider $(K+1)$ copies of the state, we denote
the inputs and outputs of Alice and each party $B_{j},$ $j\in[n-1]$
by $\chi=x_{1}...x_{K+1},$ $\upsilon_{j}=y_{j}^{1}...y_{j}^{K+1};$
$\alpha=a_{1}...a_{K+1},$ $\beta_{j}=b_{j}^{1}...b_{j}^{K+1}$ respectively.
Then, the global distribution obtained from $\ket{\Psi}^{\otimes(K+1)}$
is 
\begin{equation}
\begin{aligned}P(\alpha\beta_{1}... & \beta_{n-1}|\chi\upsilon_{1}...\upsilon_{n-1})=\\
 & P_{H}(\{a_{i}\}_{i\leq K}\{b_{j}^{i}\}_{i\leq K,j\in[n-1]}|\{x_{i}\}_{i\leq K}\{y_{j}^{i}\}_{i\leq K,j\in[n-1]})\\
 & \times P_{+}(a_{K+1}b_{1}^{K+1}...b_{n-1}^{K+1}|x_{K+1}y_{1}^{K+1}...y_{n-1}^{K+1}),
\end{aligned}
\label{eq:P+PH-copies}
\end{equation}
where $P_{H}$ comes from equation (\ref{eq:PHardy-copies}) and the
EPR2 components of $P_{H},P_{+}$ are as per equations (\ref{eq:pHlocal-copies}),
(\ref{eq:Plocal-maxent-copies}).

We now follow a similar strategy to that in Theorem \ref{thm:gmnl-from-bipartite-ent}.
To prove that the global distribution $P$ is GMNL, as is our goal,
we assume the converse, and we derive a contradiction from the nonlocality
properties of $P_{H}$ and $P_{+}.$ Assuming $P$ is not GMNL, we
can express the distribution as 
\begin{equation}
P(\alpha\beta_{1}...\beta_{n-1}|\chi\upsilon_{1}...\upsilon_{n-1})=\sum_{\lambda,M}p_{L}^{M}(\lambda)P_{M}(\alpha\{\beta_{j}\}_{j\in M}|\chi\{\upsilon_{j}\}_{j\in M},\lambda)P_{\overline{M}}(\{\beta_{j}\}_{j\in\overline{M}}|\{\upsilon_{j}\}_{j\in\overline{M}},\lambda),\label{eq:PnotGMNL-copies}
\end{equation}
where 
\begin{equation}
\sum_{\lambda,M}p_{L}^{M}(\lambda)=1,
\end{equation}
for \emph{each} $\alpha,\beta_{j},\chi,\upsilon_{j},j=1,...,n-1,$
where we recall that each $\beta_{j}=b_{j}^{1}...b_{j}^{K+1}$ and
similarly for $\upsilon_{j}.$

Now, if we sum equation (\ref{eq:PnotGMNL-copies}) over $a_{i},b_{j}^{i}$
for $i=1,...,K$ and $j=1,...,n-1$ (that is, we sum over the $i$th
digit, $i\leq K,$ of Alice and all parties $B_{j}$), we obtain $P_{+}$
on the left-hand side, from equation (\ref{eq:P+PH-copies}). On the
right-hand side, we obtain, for each $M,$ \footnote{Note that, once more, the distribution obtained by summing over only
some of the digits of a party's output still depends on the whole
input as it may be signalling in the different digits of the party's
input. However, as in Theorem \ref{thm:gmnl-from-bipartite-ent},
these extra inputs can be fixed to an arbitrary value as the left-hand
side is independent of them.} 
\begin{equation}
\sum_{\lambda}p_{L}^{M}(\lambda)P_{M}(a_{K+1}\{b_{j}^{K+1}\}_{j\in M}|\chi\{\upsilon_{j}\}_{j\in M},\lambda)P_{\overline{M}}(\{b_{j}^{K+1}\}_{j\in\overline{M}}|\{\upsilon_{j}\}_{j\in\overline{M}},\lambda)\,,\label{eq:PMMbar-P+}
\end{equation}
whose sum turns out to form an EPR2 decomposition of $P_{+}.$ Indeed,
local terms are given by bipartitions such that $M\in S_{>K},$ as
in these terms there is some digit $b_{j}^{K+1}$ with $j>K$ appearing
in $P_{\overline{M}},$ thus they are local across $A|B_{j}$ for
some $j>K.$ The nonlocal terms are given by bipartitions such that
$M\in T_{>K}$ (since all terms are nonsignalling). Therefore, the
choice of measurements which generated $P_{+}$ ensures (by equation
(\ref{eq:Plocal-maxent-copies})) that 
\begin{equation}
\sum_{\lambda,M\in S_{>K}}p_{L}^{M}(\lambda)<\varepsilon\label{eq:sumS>K}
\end{equation}
and hence 
\begin{equation}
\sum_{\lambda,M\in T_{>K}}p_{L}^{M}(\lambda)>1-\varepsilon.\label{eq:sumT>K}
\end{equation}

Going back now to equation (\ref{eq:PnotGMNL-copies}), we sum over
$a_{K+1},b_{j}^{K+1}$ for $j=1,...,n-1$ (that is, we sum over the
$(K+1)$th digit of Alice and all parties $B_{j}$). Then, we obtain
$P_{H}$ on the left-hand side, from equation (\ref{eq:P+PH-copies}).
On the right-hand side, we obtain for each $M$, 
\begin{equation}
\sum_{\lambda}p_{L}^{M}(\lambda)P_{M}(\{a_{i}\}_{i\leq K}\{b_{j}^{i}\}_{i\leq K,j\in M}|\chi\{\upsilon_{j}\}_{j\in M},\lambda)P_{\overline{M}}(\{b_{j}^{i}\}_{i\leq K,j\in\overline{M}}|\{\upsilon_{j}\}_{j\in\overline{M}},\lambda)\,,
\end{equation}
whose sum over $M$ gives an EPR2 decomposition of $P_{H}.$ This
time, $S_{\leq K}$ will give the local terms, as $P_{\overline{M}}$
will contain at least some digit $b_{j}^{j}$ for $j\leq K,$ while
$T_{\leq K}$ will give the nonlocal terms. By equation (\ref{eq:pHlocal-copies}),
our choice of $\varepsilon$ implies that 
\begin{equation}
\sum_{\lambda,M\in S_{\leq K}}p_{L}^{M}(\lambda)\leq1-\varepsilon.\label{eq:sumSleqK}
\end{equation}

Now, any bipartition in $T_{>K}$ is such that all $j\in\{K+1,...,n-1\}$
are in $M.$ Hence, there must be some $j\leq K$ in $\overline{M},$
otherwise $\overline{M}$ would be empty. Therefore, $P_{\overline{M}}$
always contains at least one digit $b_{j}^{j}$ for some $j\leq K,$
and so terms where $M\in T_{>K}$ are local across the bipartition
$A|B_{j}$ for some $j\leq K.$ That is, $T_{>K}\subseteq S_{\leq K}.$

This means that equation (\ref{eq:sumSleqK}) also holds if the sum
is over $T_{>K},$ but this is in contradiction with equation (\ref{eq:sumT>K}).
\end{proof}


\bibliography{supmat,gmnl}
\end{document}